\documentclass[draftmode]{acmsmall} % Aptara syntax

\usepackage[ruled]{algorithm2e}
\usepackage{amsmath}
\usepackage{amssymb}
\usepackage{Steven}
\renewcommand{\algorithmcfname}{ALGORITHM}
\SetAlFnt{\small}
\SetAlCapFnt{\small}
\SetAlCapNameFnt{\small}
\SetAlCapHSkip{0pt}
\IncMargin{-\parindent}

\newcommand{\ie}{{\em i.e., }}
\newcommand{\eg}{{\em e.g., }}

\newcommand {\beq} {\begin{equation}}
\newcommand {\eeq} {\end{equation}}
\newcommand {\bear} {\begin{eqnarray}}
\newcommand {\eear} {\end{eqnarray}}
\newcommand {\barr} {\begin{array}}
\newcommand {\earr} {\end{array}}

\newcommand{\Eqref}[1]{Eq.~(\ref{#1})}
\newcommand{\Eqrefs}[2]{Eqs.~(\ref{#1}) and~(\ref{#2})}

\newcommand{\fig}[1]{Fig.~\ref{#1}}

\newenvironment{itemize*}%
  {\begin{itemize}%
    \setlength{\itemsep}{0pt}%
    \setlength{\parskip}{0pt}}%
  {\end{itemize}}

\newenvironment{enumerate*}%
  {\begin{enumerate}%
    \setlength{\itemsep}{0pt}%
    \setlength{\parskip}{0pt}}%
  {\end{enumerate}}

\acmVolume{X}
\acmNumber{X}
\acmArticle{X}
\acmYear{201x}
\acmMonth{0}

\begin{document}

% Page heads
\markboth{R.~Gu\'erin et al.}{Adoption of bundled services with network externalities and correlated affinities}

% Title portion
\title{Adoption of bundled services with network externalities and correlated affinities}
\author{Roch Gu\'erin
\affil{Washington University in St. Louis}
Jaudelice C. de Oliveira
\affil{Drexel University}
Steven Weber
\affil{Drexel University}}

\begin{abstract}
%The success of an Internet technology or service is commonly measured by how widely adopted it is.  Conversely, a new Internet technology or service often becomes more attractive as its user base grows, e.g., IPv6 is more appealing as more users adopt it.  This often creates a challenging chicken-and-egg problem for new technologies and services whose small initial footprint can make it difficult to attract enough users to reach critical mass.  Bundling services is a possible solution as bundles can appeal to a larger set of early adopters, but predicting if and why this may be the case is challenging.  The goal of this paper is to develop a principled understanding of when it is beneficial to bundle technologies or services whose value is heavily dependent on the size of their user base, i.e., exhibits positive exernalities.  Of interest is how the joint distribution, and in particular the correlation, of the values users assign to components of a bundle affect its odds of success.  The results offer insight and guidelines for deciding when bundling new Internet technologies or services can help improve their overall adoption.  In particular, successful outcomes appear to require a minimum level of value correlation, but exceeding this level does not yield further benefits and can in some cases have detrimental effects.

% SW: Abstract limit of 100 words.  I edited this on the submission date 9/30/13

The goal of this paper is to develop a principled understanding of when it is beneficial to bundle technologies or services whose value is heavily dependent on the size of their user base, i.e., exhibits positive exernalities.  Of interest is how the joint distribution, and in particular the correlation, of the values users assign to components of a bundle affect its odds of success.  The results offer insight and guidelines for deciding when bundling new Internet technologies or services can help improve their overall adoption.  In particular, successful outcomes appear to require a minimum level of value correlation.
\end{abstract}

%% This is was CCS generates, but it does not seem that the
%% template has been updated to recognize it...

%% \begin{CCSXML}
%% <ccs2012>
%% <concept>
%% <concept_id>10003033.10003068.10003078</concept_id>
%% <concept_desc>Networks~Network economics</concept_desc>
%% <concept_significance>500</concept_significance>
%% </concept>
%% <concept>
%% <concept_id>10003033.10003106.10010924</concept_id>
%% <concept_desc>Networks~Public Internet</concept_desc>
%% <concept_significance>300</concept_significance>
%% </concept>
%% </ccs2012>
%% \end{CCSXML}

%% \ccsdesc[500]{Networks~Network economics}
%% \ccsdesc[300]{Networks~Public Internet}

\category{500}{Networks}{Network economics}
\category{300}{Networks}{Public Internet}

\keywords{Internet services, adoption, user valuation, correlation}

\acmformat{Roch Gu\'erin, Jaudelice C. de Oliveira, and Steven Weber,
  2013. Adoption of bundled services with network externalities and correlated affinities.}

%\begin{bottomstuff}
%This work is supported by the National Science Foundation, under grants
%\end{bottomstuff}

\maketitle

\section{Introduction}
\label{sec:intro}
Many Internet technologies, applications, and services\footnote{For
  the sake of conciseness, we use the term services in the rest of the
  paper.} have value that increases with the size of their user base,
\ie they exhibit positive externalities or network effects.
%% (Metcalfe's law is often mentioned as one of the first
%% acknowledgments of this effect in modern communication
%% networks~\cite[p.71]{Castells2010}). 
Externalities are well-known~\cite{Cabral1990} to have
dual effects on the adoption pattern of services.  Adoption rapidly
accelerates after passing a critical threshold (until the market
starts saturating), but reaching this critical level of adoption is
often slow and difficult.  In practice, services that fail often do so
during this early stage, as many potential adopters see a cost that
exceeds the (low) initial value of the service.  This is commonly
mentioned as an explanation for the limited or stalled adoption
of many Internet security protocols~\cite{OzmSch06}.

A common approach (see again~\cite{OzmSch06}) to overcome this initial
hurdle is to bundle services, in the hope that the bundle has broader
appeal and is, therefore, able to overcome early adoption inertia.
The main unknown is the extent to which dependencies (as measured by
a joint distribution) or \emph{correlation} in how
users value the individual services influence their adoption decision
for the bundle.
%% For example, combining two reasonably popular services may be
%% detrimental to both, if they appeal to distinct population segments
%% (the bundle's adoption cost is likely to be higher than that of
%% either service, and individual users may not perceive a higher
%% aggregate value).  Conversely, combining two middling services that
%% appeal to similar though not identical sets of users might attract
%% more early adopters, and allow the bundle to build critical mass
%% and realize a higher level of adoption than that of either service
%% offered independently.  The possibility of such very different
%% outcomes calls for developing a better understanding for when each
%% may arise.  Exploring this issue is the main focus of the paper.
%% Of particular interest is how adoption decisions for either
%% individual services or a bundle are affected by \emph{correlation}
%% in the value users assign to each service.  
We illustrate this next by way of examples, which also (and
more importantly) demonstrate the diversity of Internet services
for which this arises.

%% The investigation relies on simple models that yield insight and
%% guidelines for deciding when services can benefit from bundling.
%% However, before reviewing those results, we introduce a few
%% representative examples that illustrate the range of settings under
%% which such questions arise.  The intent is not to offer a detailed
%% step-by-step description of how models' parameters map to specific
%% services.  Instead, the goal is to highlight the applicability of
%% the paper's findings across a wide range of Internet services.
%%\begin{description}
%%\item[Online discussion forums]

\subsection{Anonymous communications and secure distributed storage}  

Anonymous communication systems have been available for some time, \eg
see~\cite{franco2012} for a recent survey, but in spite of a recent
rise in both profile~\cite{arthur2013} and usage~\cite{brewster2013},
they remain relatively marginal, \ie have not yet attracted a large
user-base.  This can affect their robustness and their ability to
deliver strong anonymity guarantees (mixing traffic from more users
and tapping into the resources contributed by those users can improve
both anonymity and robustness, at least in P2P based implementations
of such systems).

Overcoming the limited appeal (to users) of anonymous communications
and increasing the number of users such a system can tap into, can be
realized by bundling it with another service.  Ideally, this other
service should exhibit technical synergies with anonymous
communications so as to facilitate a joint implementation.  Secure
distributed storage is a possible candidate.  It enables the automatic
and encrypted backup of local files over a distributed set of network
peers (see BuddyBackup\footnote{http://www.buddybackup.com/.} for an
example), and shares with anonymous communications a reliance on
cryptographic primitives and protocols, as well as a value that grows
with its number of adopters (more users likely means a more reliable
system).  The main question is whether combining those two services
can increase adoption for both.  The answer depends on the cost versus
value of the bundle, and how this varies across users.

The cost of the bundle consists of the communication (bandwidth),
processing, and storage costs of the two services, with anonymous
communications calling mostly for bandwidth and processing resources,
and secure distributed storage requiring primarily storage resources
and to a lesser extent processing and communication resources.
Because the two services have mostly independent needs, those costs
should be approximately additive.  The value a user assigns to the
bundle depends on her level of usage of anonymous communications and
reliance on secure distributed storage as a means of preserving and
accessing her personal data.  This value will change as more users
adopt the bundle (it improves the quality and reliability of both
services), but the decisions of early adopters depend primarily on
how they intrinsically value access to anonymous communication and
secure distributed storage.

For illustration purposes, assume that within a given user population
the stand-alone values of both services are uniformly distributed.
However, to reflect the fact that secure distributed storage should be
attractive to most users while anonymous communications will likely
have more limited appeal, we assume that the stand-alone values of the
former are in $[a, 1], 0<a<1,$ while they span the full $[0,1]$ range
for the latter.  In other words, most users view secure distributed
storage as useful (valued at $\geq a$), while fewer assign a similar
value (in the range $[a,1]$) to anonymous communications.  Under those
assumptions, correlation in user valuations clearly affects the number
of early adopters the service bundle will attract.  For example, it is
relatively easy to show (see Section~\ref{sec:conaff} for related
derivation details) that the cost threshold beyond which there are no
early adopters for the bundle is~$2$ under perfect positive
correlation, but only $a+1$ under perfect negative correlation.

%%\begin{description}
%%\item[Online discussion forums]
%% \smallskip
%% \noindent
\subsection{Online discussion forums}

Consider next the case of an online
discussion forum\footnote{Similar arguments hold for other systems of
  a similar ``crowdsourcing'' nature, \eg recommender systems.} dedicated
to a particular topic.  Participating in such a forum has some
intrinsic value, \eg from access to promotions and discounts on
related products, but its core value often comes from the answers and
advice it provides in response to users' questions.  To succeed, a
forum must, therefore, accumulate a large enough ``knowledge-base''
and consequently achieve a critical mass of users.  This can be
challenging, as the added value from Q\&A's is essentially absent in
the early stages, and promotions and discounts alone may be
insufficient to attract enough early adopters.  Combining the topics
of multiple forums under a common umbrella is one way to address this
challenge.  The stand-alone value of such a ``bundled'' forum, \eg
promotions and discounts that now extend across more products, may
appeal to a broader user base, and allow it to succeed where
individual forums would not have.  The question we seek to
answer is again when and why this may be the case?

As with anonymous communications and secure distributed storage,
whether a bundled forum attracts more early adopters, and therefore
improves its odds of success, depends on its initial cost-benefit
ratio relative to that of individual forums.  The ``cost'' of joining
a bundled forum, \eg the amount of time needed to extract useful
information, can be higher than that of more focused,
single-topic forums.  Its combined stand-alone value arguably depends
on many factors, but a reasonable first approximation is again to
assume that it is the sum of the stand-alone values (product
promotions and discounts) associated with both topics.  As in the
previous example, whether this sum exceeds the cost of joining the
forum, which determines the number of early adopters, depends to a
large extent on the \emph{joint distribution} of user valuations for
the individual forums; an important measure of which is their
correlation.

For purpose of illustration, consider a scenario where we contemplate
merging two discussion forums, whose stand-alone values follow
identical uniform distributions when measured across a population of
users (they are of equal value on average).  Assume further that for a
given user, the values she sees in the two forums are either perfectly
positively or perfectly negatively correlated, \ie equal or
diametrically opposed.  Under perfect positive correlation, the
stand-alone value that any user derives from the bundled forum is then
simply \emph{twice} the value she sees in either individual
forum.  If we assume that the cost of joining the bundled forum is
also twice that of joining a single-topic forum, \eg it takes twice as
long to find relevant information, then it is easy to show that
bundling has no impact on early adoption, and the bundled forum sees
the same number of potential early adopters\footnote{Note though that
  final adoption can be different depending on the strength of the
  externality factor of the combined topics.} as either 
original forum.  In contrast, when values are perfectly negatively
correlated, all users now see the \emph{same} (average) value from
joining the bundled forum.  In this case either no user or all users
will be early adopters, depending on whether this average value is
above or below the bundle's cost.  Hence, unlike the case of perfect
positive correlation, bundling can significantly affect
the number of early adopters.

\subsection{Summary}
As the above examples hopefully illustrated, correlation in how users
value different services and/or technologies (and more generally their
joint distribution) can have a significant effect on whether combining
them in a bundle is beneficial.  In the rest of this paper, we explore
this issue in a systematic fashion.  Section~\ref{sec:related} offers
a brief review of related works.  Section~\ref{sec:matmod} introduces
our model for service adoption.  Section~\ref{sec:conaff} considers
the case where user affinities for the services are represented as
continuous uniform random variables.  It first explores the extreme
cases of perfect positive and negative correlation in
Section~\ref{ssec:perposcor} and Section~\ref{ssec:pernegcor},
respectively, while Section~\ref{ssec:uncor} investigates the
intermediate case of independent affinities.  The latter section
illustrates the difficulty of characterizing the bundle's adoption
under general correlation, and motivates the model of
Section~\ref{sec:disaff}
%% explores intermediate correlation scenarios using a model 
where user affinities for the services are captured through correlated
discrete (Bernoulli) random variables with parameterized correlation.
Section~\ref{sec:guidelines} articulates the findings that emerge from
the analysis, and numerically explores their robustness through
limited extensions to the model. Section~\ref{sec:con} offers a brief
conclusion.

\section{Related work}
\label{sec:related}

The topic of this paper is at the intersection of two major lines of
work; product and technology diffusion, and product and service
bundling.

Modeling how products and services diffuse through a population of
potential users, \ie are being adopted, is a topic of longstanding
interest in marketing research with \cite{PerMul2010} offering a
recent review of available models and techniques.  The models most
relevant to our investigation are those based on the approach
introduced in~\cite{Cabral1990} and extended in many subsequent works,
which explore product diffusion in the presence of externalities using
an adoption decision process that reflects the utility of individual
users.  As described in Section~\ref{sec:matmod}, the adoption model we
rely on belongs to this line of work.  However, and except for a few
recent works that we review below, the aspect of bundling had not been
incorporated in those investigations, and this is one of the aspects
we focus on in the paper.

There has obviously been a significant literature devoted to bundling
as a stand-alone topic (see~\cite{VenMah2009} for a recent review).
The main goal of most of those works has typically been the
development of optimal bundling and pricing strategies, and pricing is
a dimension that is largely absent from our investigation in that
service costs\footnote{The time needed to retrieve information in a
  discussion forum, or the communication, processing and storage
  resources that either anonymous communications or secure distributed
  storage require.} are assumed given and exogenous.  Instead, our
focus is mostly on how the joint distribution in service valuation
across users (as measured through their correlation coefficient),
determines whether the adoption level of a service bundle can exceed
those of separate service offerings.  Correlation in how users value
different services and the impact this has on bundling strategies is
in itself a topic that several prior works have explicitly taken into
account,
\eg\cite{Schmalensee1984,McAfeeMcMillanWhinston1989,BakosBrynjolfsson1999}.
In general, negative correlation in demand improves bundling's
benefits over separate offerings, although high marginal costs
(compared to the average value of the bundle) can negate this effect.
Conversely, a high positive correlation tends to yield the opposite
outcome, \ie favor separate (pure component) offerings.  As
highlighted in the examples of Section~\ref{sec:intro}, our focus on
maximizing adoption results in more nuanced outcomes, with correlation
playing a more ambivalent role in determining the success of a bundled
offering.  Furthermore, early works on bundling did not account for
the potential impact of externalities.

There are to-date only three works we are aware
of~\cite{PraVen2010,PanEtz2012,ChaoDerden2013} that have investigated
the problem of bundling products or services with externalities, and
we briefly review how these papers differ from our investigation.
First and as has been the norm in the bundling literature, those three
papers all focus on optimal pricing strategies, while we assume that
costs (prices) are exogenous parameters and instead seek to understand
their impact (and that of other factors) on the adoption level of a
bundled offering.  Second, the impact of value (demand) correlation is
essentially absent from those three prior works.  

Specifically, \cite{PraVen2010} focuses on optimal pricing while
assuming \emph{independent} valuations for the two services.
\cite{PanEtz2012} explores the joint offering of a product and a
complementary service, where the latter exhibits positive
externalities.  As in~\cite{PraVen2010}, users' valuations for the
product and its complementary service are assumed independent, and
there is no investigation of the potential impact of their
correlation.  \cite{ChaoDerden2013} is cast in the context of a
two-sided market (the two market sides create externalities), where
the platform provider seeks to decide how to bundle and price new
content with the platform it offers, given the existence of an
installed based of users and content developers.  The focus is again
on optimal pricing strategies and bundling decisions, and there is no
correlation between the value of the new content and that of earlier
content.  These are again important differences with our work, which
furthermore does not involve the presence of an existing user base.
As~\cite{ChaoDerden2013}'s title indicates and as the paper emphasizes
(Section~2.2), this plays an important role in the platform's
strategic decisions.  In contrast, our interest lies primarily in
understanding how bundling can help nascent Internet services
ultimately reach the high levels of adoption they need to realize
their full value potential and succeed.

\section{Mathematical model}
\label{sec:matmod}
In this section, we review the basic structure of the models on which
we rely to capture the evolution of service adoption.
\subsection{Separate (unbundled) service offerings}
\label{ssec:sepseroff}

We consider a model for the adoption of multiple (two) services by a
heterogeneous population of potential users.  The perceived utility
$V_i(x_i(t))$ of service $i\in\{1,2\}$ by a (random) user given that a
fraction $x_i(t) \in [0,1]$ of the population has adopted the service
at time $t$ incorporates three components: $i)$ the user's affinity
(stand-alone value) for the service (capturing users' heterogeneity),
$ii)$ the network externality tied to the adoption level of the
service, and $iii)$ the service cost.  Specifically: 
\begin{equation}
\label{eq:a}
V_i(x_i(t)) = U_i + e_i x_i(t) - c_i, ~ i \in \{1,2\},
\end{equation}
where $i)$ $U_i \geq 0$ is the user's (random) affinity for service
$i$; $ii)$ $e_i \geq 0$ is the strength of the externality
contribution for service $i$; and $iii)$ $c_i \geq 0$ is the cost of
adopting service $i$. As is common, for analytical tractability we
adopt a linear externality model\footnote{The assumption of linear
  externality typically does not affect the nature of the findings,
  \ie they hold qualitatively for distributions with CDF $F(u)$, which
  share with the uniform distribution a non-decreasing hazard-rate
  function
  $F^{\prime}(u)/\left(1-F(u)\right)$~\cite{Bhargava,Fudenberg}.}.

When services are offered separately, users make adoption decisions
for each based on their respective utilities: 
\begin{equation*}
\mbox{user adopts service $i$ at time $t$ with adoption level
  $x_i(t)$} ~ \Leftrightarrow ~ V_i(x_i(t)) > 0.  
\end{equation*}
In particular, there is no ``budget constraint'' where adoption of one
service by a user affects adoption of the other service by that user;
this is natural given our focus on adoption costs, \eg communication,
storage, processing, etc., rather than pricing.  However, while
service adoption decisions are uncoupled, the random variables
$(U_1,U_2)$, capturing heterogeneity in user affinity, may be
correlated. 

Denote as $h_i(x_i) = \Pbb(V_i(x_i) > 0)$ the probability a random
user adopts service $i$ given an adoption level $x_i$. 
%% The diffusion model for the adoption process of service $i$ is
%% \begin{equation}
%% \dot{x}_i(t) = \gamma (h_i(x_i(t)) - x_i(t)) 
%% \end{equation}
%% and asserts the rate of increase of adoption at time $t$ is
%% proportional (with proportionality constant $\gamma > 0$) to the
%% discrepancy between the fraction of users that would like to adopt at
%% adoption level $x_i(t)$, \ie $h_i(x_i(t))$, and the fraction of users
%% that have actually adopted, $x_i(t)$.  
An equilibrium for this model is any $x_i^*$ such that
%% $\dot{x}_i^*(t) = 0$, \ie
\begin{equation}
\label{eq:equil_s}
h_i(x_i^*) = \Pbb(U_i > c_i - e_i x_i^*) = x_i^*.
\end{equation}
When the two services are offered separately they achieve adoption
equilibria $(x_1^*,x_2^*) \in [0,1]\times [0,1]$.  One of our goals is
to compare these equilibria to those realized when the two services
are bundled, and characterize differences as a function of the model's
parameters $(e_1,e_2,c_1,c_2)$ and the joint distribution of
$(U_1,U_2)$.  

\subsection{Bundled service offerings}
\label{ssec:bunseroff}

Under bundling, a user must decide whether to adopt either both
services or neither, \ie we do not consider the case of mixed
bundling where services are simultaneously offered as a bundle and
separately. The basis for a user's decision is now the aggregate
utility she derives from the bundle:  
\begin{equation*}
\mbox{user adopts the bundle at time $t$} ~ \Leftrightarrow ~ V(x(t)) > 0,
\end{equation*}
where consistent with~\Eqref{eq:a}
\begin{equation}
\label{eq:b}
V(x(t)) = V_1(x(t)) + V_2(x(t)) = U + e x(t) - c.
\end{equation}
Here, $x(t)$ is the (common) adoption level of the bundled services.
Note the assumption of additive values for the two services, \ie
$V(x(t)) = V_1(x(t)) + V_2(x(t))$, which implies that they are neither
substitute nor complement.  Under this assumption, $U = U_1 + U_2$ is
the aggregate intrinsic value of the bundled services, $e = e_1+e_2$
is the aggregate force of the externality, and $c$ is the aggregate
cost, which for simplicity also satisfies $c=c_1+c_2$.  Extending the
model to account for instances where the two services are partial
complements or substitutes, as well as for possible economies of scope
in the cost of a service bundle is clearly of interest.  Such
extensions can be readily incorporated in the models, but add
complexity.  Furthermore, while they quantitatively affect adoption
outcomes, \ie if and when bundling is beneficial, the qualitative
findings of the paper still hold, and in particular the importance of
the joint distribution (correlation) of service affinities in the
efficacy of bundling.

% SW (Fri 6/28/13): this seems an important point to discuss at
% greater length.

As with separate offerings, equilibria satisfy 
\begin{equation}
\label{eq:BunCCDF}
h(x^*) = \Pbb(U > c - e x^*) = x^*.
\end{equation}
Our question can now be restated more formally: {\em how do adoption
  equilibria $(x^*,x^*)$ under bundling compare with adoption
  equilibria $(x_1^*,x_2^*)$ when services are offered separately?} 

It remains to specify the joint distribution of service
affinities (values a user assigns to each service) $(U_1,U_2)$.  In
Section~\ref{sec:conaff} we consider the case 
where $(U_1,U_2)$ are uniform continuous random variables, while in
Section~\ref{sec:disaff} we assume that $(U_1,U_2)$ are uniform discrete
(Bernoulli) random variables.  Table~\ref{tab:notation} lists commonly
used notation.    

\begin{table}
\tbl{Notation\label{tab:notation}}{
\begin{tabular}{rl} \hline 
$x_i(t)$ & adoption level of (unbundled) service $i$ at time $t$ \\
$x(t)$ & adoption level of service bundle at time $t$ \\
$x_i^*,x^*$ & equilibria adoption level for service $i$ and bundle \\
$(U_1,U_2)$ & random (unbundled) service affinities \\
$U = U_1+U_2$ & affinity for service bundle \\
$(e_1,e_2)$ & externality for (unbundled) services $1,2$ \\
$e = e_1 + e_2$ & externality for service bundle \\
$(c_1,c_2)$ & costs for adopting (unbundled) services $1,2$ \\
$c = c_1 + c_2$ & costs for adopting service bundle \\
$(V_1,V_2)$ & utility function for (unbundled) services $1,2$ \\
$V = V_1+V_2$ & utility function for service bundle \\
$h_i(x_i) = \Pbb(V_i(x_i) > 0)$ & probability of adoption of (unbundled) services $1,2$ \\
$h(x) = \Pbb(V(x) > 0)$ & probability of adoption of service bundle \\
$l_i = \frac{c_i-1}{e_i}$ & left adoption threshold for (unbundled) services $1,2$ \\
$r_i = \frac{c_i}{e_i}$ & right adoption threshold for (unbundled) services $1,2$ \\
$l = \frac{c-2}{e}$ & left adoption threshold for service bundle \\
$m = \frac{c-1}{e}$ & middle adoption threshold for service bundle \\
$r = \frac{c}{e}$ & right adoption threshold for service bundle \\
$\rho$ & correlation parameter for $(U_1,U_2)$ in \eqref{eq:c} \\
$p$ & distribution parameter for $(U_1,U_2)$ in \eqref{eq:disaffjointdisbn} \\
$n$ & population size for Monte-Carlo simulations \\ \hline
\end{tabular}}
\end{table} 

\section{Continuous affinities}
\label{sec:conaff}

In this section, we assume $(U_1,U_2)$ are continuous uniform random
variables on $[0,1]$.  This assumption is sufficient to characterize
the equilibria under separate service offerings, which we do in
Section~\ref{ssec:conaffsep}.  The equilibria under bundled service
offerings, however, depend upon the correlation between $(U_1,U_2)$,
since the bundled affinity depends upon $U = U_1 + U_2$.  A
distribution on $(U_1,U_2)$ with a parameterized correlation is
presented in Section~\ref{ssec:gencor}, but is difficult to work with
in its general form\footnote{See~\cite[Section 2]{VenMah2009} for a
  related discussion on the difficulty of using uniform distributions
  to study the impact of correlation on bundling.}.  Consequently, we
investigate the three special cases where $(U_1,U_2)$ are perfectly
positively (Section~\ref{ssec:perposcor}) and negatively
(Section~\ref{ssec:pernegcor}) correlated, and independent
(Section~\ref{ssec:uncor}).  The intent is to illustrate that
different types of outcomes arise at the two correlation extremes, and
confirm the difficulties associated with capturing the impact of
intermediate correlation values when using continuous distributions.
This then motivates the more tractable discrete model of
Section~\ref{sec:disaff}, wherein we are able to more explicitly
explore the impact of varying correlation.  Numerical investigations
are then used in Section~\ref{sec:impact_rho_con} to verify that
findings obtained with the simpler discrete model also hold under the
continuous model.

\subsection{Separate offerings}
\label{ssec:conaffsep}

The following proposition characterizes the possible equilibria for
separate service offerings when the user service affinities are
uniform random variables. 
%% The results are illustrated in Fig.\ \ref{fig:conaffsep1}.

\begin{proposition}
\label{prop:conaffsepoffequ}
When $U_i$ ($i \in \{1,2\}$) is uniformly distributed on $[0,1]$, the
probability of user adoption of service $i$ in \Eqref{eq:equil_s}
becomes 
\begin{equation}
h_i(x_i) = \left\{ \begin{array}{lrcccl}
0, \; & & & x_i & \leq & l_i \\
e_i x_i + 1 - c_i, \; & l_i & < & x_i & \leq & r_i \\
1, \; & r_i & > & x_i & & 
\end{array} \right. 
\end{equation}
for adoption thresholds $l_i \equiv \frac{c_i - 1}{e_i}$ and $r_i \equiv \frac{c_i}{e_i}$.  The three possible equilibria are $x_i^* \in \{0,(1-c_i)/(1-e_i),1\}$.  The conditions for each equilibrium (eq.) are:
\begin{eqnarray*}
0 \mbox{ is stable eq.} & \Leftrightarrow & c_i > 1 \nonumber \\
\frac{1-c_i}{1-e_i} \mbox{ is stable eq.} & \Leftrightarrow & e_i < c_i < 1 \nonumber \\
\frac{1-c_i}{1-e_i} \mbox{ is unstable eq.} & \Leftrightarrow & 1 < e_i < c_i \nonumber \\
1 \mbox{ is stable eq.} & \Leftrightarrow & e_i > c_i \label{eq:conaffsepeq1}
\end{eqnarray*}
The lowest stable equilibrium (lseq.) adoption level for each $(c_i,e_i)$ is
\begin{eqnarray*}
0 \mbox{ is lseq.} & \Leftrightarrow & c_i > 1 \nonumber \\
\frac{1-c_i}{1-e_i} \mbox{ is lseq.} & \Leftrightarrow & e_i < c_i < 1 \nonumber \\
1 \mbox{ is lseq.} & \Leftrightarrow & c_i < \min\{e_i,1\} \label{eq:conaffsepeq2}
\end{eqnarray*}
\end{proposition}

The proof of Prop.\ \ref{prop:conaffsepoffequ} is straightforward and
is omitted. 
% SW (Tue 7/2/13): is this okay? Seems tedious to write it all out.
Note that the equilibria conditions
%% in \Eqref{eq:conaffsepeq1} 
are a cover but not a partition of the $(c_i,e_i)$ plane, while the
%%equilibria 
lseq.\ conditions\
%% in \Eqref{eq:conaffsepeq2} 
are a partition of the $(c_i,e_i)$ plane (see \fig{fig:conaffsep1} in
the appendix for an illustration).
%%, as illustrated in the right side of Fig.\ \ref{fig:conaffsep1}.  
Note also that the notion of lowest stable equilibrium $(lseq.)$ is
natural is our setting, where we consider services that have an
initial adoption level of $0$ when first offered, \ie $x_i(0) = 0$, so
that the lseq. will be the achieved equilibrium.  The conditions on
the equilibria are also intuitive: zero adoption results when the
costs are high, full adoption results when the externality effect
outweighs the cost, and partial adoption results when costs are low
but outweigh the externality effect.

\subsection{Bundling under general correlation}
\label{ssec:gencor}

The equilibria under bundled service offerings with continuous uniform
affinities $(U_1,U_2)$ depend upon the correlation between them.
There are many ways to generate random variables with parametrized
correlation.  We rely on a standard approach~\cite{Pea1907,HotPab1936}
(see the Appendix, Section~\ref{sec:gen_corr_var} for details) to
generate a pair of uniform random 
variables $(U_1,U_2)$ with correlation coefficient $\rho$ for any
value $\rho\in[-1,1]$.  The approach uses a pair of independent
standard normal random variables as its starting point, so that the
joint distribution $F_{U_1,U_2}(u_1,u_2)$, the distribution of the
aggregate service affinity $F_U(u)$ for $U = U_1 + U_2$, and the
resulting probability of adoption $h(x)$, can all be described in
terms of the standard normal CDF
%% and PDF, 
$F_Z$.
%% and $f_Z$, respectively.

As seen in the Appendix, the resulting expressions are,
in general, rather unwieldy, and for illustration purposes, we restate
below the expression for $h(x)$ that can be used to determine adoption
equilibria.  
%% The difficulty of characterizing equilibria for general correlation
%% coefficient is one of the motivations for considering the special
%% cases $\rho \in \{-1,0,+1\}$ in subsections. 
{\small
\begin{equation}
h(x) = \left\{ \begin{array}{lrcccl}
0, \; & & & x & \leq & l \\
2 - (c-ex) - \displaystyle\int_{c-ex-1}^1 F_Z \left( \frac{F_Z^{-1}(c-ex-v) - \rho F_Z^{-1}(v)}{\sqrt{1-\rho^2}} \right) \drm v , \; & l & < & x & \leq & m \\
1 - \displaystyle\int_0^{c-ex} F_Z \left( \frac{F_Z^{-1}(c-ex-v) - \rho F_Z^{-1}(v)}{\sqrt{1-\rho^2}} \right) \drm v, \; & m & < & x & \leq & r \\
1, \; & r & < & x & & 
\end{array} \right.
\label{eq:conaffbunprobadopt-main}
\end{equation}
} 
for adoption thresholds $l \equiv \frac{c-2}{e}$, $m \equiv \frac{c-1}{e}$, and $r \equiv \frac{c}{e}$.  

The equilibria under bundling are the solutions of $h(x) = x$.  As
evident from \Eqref{eq:conaffbunprobadopt-main}, this is a difficult
equation to work with. 
%% Although $h'(x) \geq 0$ is an immediate property, a characterization of $h^{''}(x)$ depends upon the derivative of $f_U(u)$ in \eqref{eq:conaffsumpdf}, which is a complicated expression.  These difficulties 
This motivates focusing on the three specific cases of perfect
positive ($U_1=U_2$) and negative ($U_1=1-U_2$) correlation, as well
as independence, \ie $\rho \in \{-1,0,+1\}$.

\subsection{Perfect positive correlation}
\label{ssec:perposcor}

Specializing for $\rho = 1$ Props.\ \ref{prop:jointcdfpdfcorrcont}
and~\ref{prop:conaffsumpdfcdf} of the Appendix yields the 
joint and sum distributions for uniform random variables $(U_1,U_2)$
satisfying $U_1 = U_2$, and thus $U = U_1 + U_2$ is uniform over
$[0,2]$: 
\begin{eqnarray*}
F_{U_1,U_2}(u_1,u_2) = \min\{u_1,u_2\} & & f_{U_1,U_2}(u_1,u_2) = \mathbf{1}_{u_1=u_2} \nonumber \\
F_U(u) = u/2, ~ u \in [0,2] & & 
f_U(u) = 1/2, ~ u \in [0,2]. 
\label{eq:conaffposcdfpdf}
\end{eqnarray*}
Because the aggregate affinity is uniformly distributed on $[0,2]$,
the resulting equilibria are of the same form as in
Prop.\ \ref{prop:conaffsepoffequ} after replacing $e_i$ and $c_i$ by
$e/2$ and $c/2$, respectively.  Thus we have the following corollary
to Prop.\ \ref{prop:conaffbunprobadopt} and
Prop.\ \ref{prop:conaffsepoffequ}. 
\begin{corollary}
\label{cor:conaffposprobadopt}
The probability of bundle adoption $h(x)$ in
Prop.\ \ref{prop:conaffbunprobadopt} under aggregate affinity
$U=U_1+U_2$ formed from perfectly positively correlated uniform
affinities $(U_1,U_2)$ satisfying $U_1 = U_2$ is 
\begin{equation}
h(x) = \left\{ \begin{array}{lrcccl}
0, \; & & & x & \leq & l \\
\frac{e}{2} x + 1 - \frac{c}{2}, \; & l & < & x & \leq & r \\
1, \; & r & < & x & & 
\end{array} \right. 
\end{equation}
The three possible equilibria are $x^* \in \{0,(2-c)/(2-e),1\}$.  The
conditions for each equilibrium (eq.) are: 
\begin{eqnarray*}
&& 0 \mbox{ is stable eq.} \Leftrightarrow c > 2, 
\quad  1 \mbox{ is stable eq.} \Leftrightarrow e > c \nonumber \\
&& \frac{2-c}{2-e} \mbox{ is stable eq.} \Leftrightarrow e < c < 2, \quad 
%%\nonumber \\
\frac{2-c}{2-e} \mbox{ is unstable eq.} \Leftrightarrow 2 < e < c
%%\\
%% \nonumber \\
%%&& 1 \mbox{ is stable eq.} \Leftrightarrow e > c \nonumber
\label{eq:conaffposeq1}
\end{eqnarray*}
The lowest stable equilibrium (lseq.) adoption level for each $(c,e)$ is
\begin{equation*}
0 \mbox{ is lseq.} \Leftrightarrow c > 2 , \quad
\frac{2-c}{2-e} \mbox{ is lseq.} \Leftrightarrow  e < c < 2, \quad
1 \mbox{ is lseq.} \Leftrightarrow c < \min\{e,2\} \label{eq:conaffposeq2}
\end{equation*}
\end{corollary}

The next part of the analysis is to compare the lowest stable
equilibria under separate ($(x_1^*,x_2^*)$) and bundled ($(x^*,x^*)$)
offerings as a function of the system parameters $(e_1,e_2,c_1,c_2)$.
The results are shown below.   
{\small
\begin{equation*}
\begin{array}{c|cc||ccc}
& & & 0 & \frac{2-c}{2-e} & 1 \\ \hline
& & & c > 2 & e<c<2 & c < e \land 2  \\ \hline \hline
% Row 1
(0,0) & c_1 > 1 & c_2 > 1 & SS & WW & WW \\
& & & \mbox{True} & \mbox{False} & \mbox{False} \\ \hline
% Row 2
\left(0,\frac{1-c_2}{1-e_2} \right) & c_1 > 1 & e_2 < c_2 < 1 & SL &
WL \mbox{ or } WW & WW \\ \hline
% Row 3
(0,1) & c_1 > 1 & c_2 < e_2 \land 1 & SL & WL & WS \\ \hline 
% Row 4
\left(\frac{1-c_1}{1-e_1},0\right) & e_1 < c_1 < 1 & c_2 > 1 & LS & LW \mbox{ or } WW & WW \\ \hline
% Row 5
(1,0) & c_1 < e_1 \land 1 & c_2 > 1 & LS & LW & SW \\ \hline
% Row 6
\left(\frac{1-c_1}{1-e_1},\frac{1-c_2}{1-e_2}\right) & e_1 < c_1 < 1 & e_2 < c_2 < 1 & LL & WL \mbox{ or } LW & WW \\
& & & \mbox{False} & & \mbox{False} \\ \hline
% Row 7
\left(\frac{1-c_1}{1-e_1},1 \right) & e_1 < c_1 < 1 & c_2 < e_2 \land 1 & LL & WL & WS \\
& & & \mbox{False} & & \\ \hline
% Row 8 
\left(1,\frac{1-c_2}{1-e_2} \right) & c_1 < e_1 \land 1 & e_2 < c_2 < 1 & LL & LW & SW \\
& & & \mbox{False} & & \\ \hline
% Row 9
(1,1) & c_1 < e_1 \land 1 & c_2 < e_2 \land 1 & LL & LL & SS \\
& & & \mbox{False} & \mbox{False} & \mbox{True} \\ \hline
\end{array}
\end{equation*}
}
The nine rows are lowest stable equilibria $(x_1^*,x_2^*)$ under
separate offerings, with $x_i^* \in \{0,(1-c_i)/(1-e_i),1\}$ for $i
\in \{1,2\}$.  The third column corresponds to the lowest stable
equilibria $(x^*,x^*)$ under bundling with $x^* \in
\{0,(2-c)/(2-e),1\}$.  Each combination of row and third column entry,
say $(x_1^*,x_2^*,x^*)$, is a possible equilibrium triple without and
with bundling.  The second column gives the conditions on
$(c_1,c_2,e_1,e_2)$ for each $(x_1^*,x_2^*)$ to be the lowest stable
equilibria under separate offerings, and the second row of the third
column gives the conditions on $(c,e)$ for each $x^*$ to be the lowest
stable equilibria under bundling.

Each third column entry is labeled with a pair of letters
$(\Delta_1,\Delta_2)$ with $\Delta_i \in \{L,S,W\}$ for $i \in
\{1,2\}$ representing (L)ose, (S)ame, and (W)in, and denoting the
change in equilibrium under bundling for that service.  For example,
$SL$ means the equilibrium for service~$1$ stayed the same ($x_1^* =
x^*$), while the equilibrium for service~$2$ decreased ($x_2^* >
x^*$).  The notation $a\land b$ in the inequalities in the second
column, simply means that the inequality needs to be satisfied for
both $a$ ``AND'' $b$.  The word ``True'' indicates the equilibrium for the
column always results for the equilibria in the corresponding row, \eg
when $c_1 > 1$ and $c_2 > 1$, the bundled equilibrium $0$ always result
for the separate equilibria $(0,0)$ because the conditions on $c_1$
and $c_2$ imply $c > 2$.  Conversely, the word ``False'' indicates the
equilibrium for the column is never feasible for the equilibria in the
corresponding row. %% Observe the entry for
%% $((1-c_1)/(1-e_1),(1-c_2)/(1-e_2),(2-c)/(2-e))$ is $WL$ or $LW$ and
%% the other entries are ``False'', \ie when both adoption levels under
%% separate offerings are neither $0$ nor $1$ then the bundled adoption
%% level is also neither $0$ nor $1$, and the equilibrium level lies
%% between the two separate offering adoption levels.  
%% RG:  Updated on 8/18/13

There are nine possible tuples.  Under perfectly positively
correlated user valuations, the bundle's valuation is essentially a
weighted sum of the valuations of the individual services, so that
most outcomes involve a trade-off between improving (or maintaining)
the adoption of one service and worsening (or maintaining) that of the
other.  Of note is the fact that a $LL$ outcome is not feasible.  This
is because, a bundle equilibrium of $0$ only arises when the less
valuable service also has an equilibrium of $0$ when offered alone,
which results in a $SL$ (or $LS$) outcome.  Because of the effect of
externalities, the converse is, however not true, \ie $WW$ outcomes
can be realized.  

$WW$ outcomes typically arise when one technology has a high adoption
cost together with a high externality factor, while the other
technology enjoys middling cost and externality factor.  In such
cases, the first technology could be tremendously successful, if only
it managed to acquire enough of a user base to unleash the
value its strong externality factor can deliver.  However, its high
adoption cost makes this nearly impossible.  Hence, when offered
alone, this technology never takes off.  In contrast, the relatively
low adoption cost of the other technology enables it to make rapid
initial progress even when offered alone.  Its initial adoption spurt,
however, quickly subsides as its externality contributions do not
progress fast enough to keep attracting more users.  This translates
in neither technology experiencing meaningful success when offered
alone.  Bundling can, however, change this.

When the two technologies are bundled, the second becomes the engine
that drives initial adoption until enough of a user-base has been
built to allow the first technology to cross its critical adoption
threshold.  At that point, the roles reverse and the first technology
becomes the main driver for continued adoption, as its strong
externality contribution is now sufficient to attract more users.  The
bundle's adoption then takes off, possibly reaching full penetration.
In the process, the second technology also reaches a level of
adoption it would never have realized on its own.
%% Several entries have multiple possible change pairs, \eg $WL$ or $WW$
%% for $(0,(1-c_2)/(1-e_2),(2-c)/(2-e))$.  This simply means that the
%% service $2$ adoption level decreases ($L$) /increases ($W$) if $x_2^*
%% \gtrless x^*$, respectively, \ie $(1-c_2)/(1-e_2) \gtrless
%% (2-c)/(2-e)$.  
%%
%% Finally, eight of the nine possible change tuples
%% $(\Delta_1,\Delta_2)$ are represented in the table -- the only missing
%% possibility is $LL$, \ie bundling can never cause a decrease in the
%% adoption level of both services when the service affinities are
%% perfectly positively correlated.  At worst, bundling services in this
%% case can cause one adoption level to stay the same and the other to
%% decrease, and at best both adoption levels can increase.
%%
\subsection{Perfect negative correlation}
\label{ssec:pernegcor}

Specializing for $\rho = -1$ Props.~\ref{prop:jointcdfpdfcorrcont}
and~\ \ref{prop:conaffsumpdfcdf} of the Appendix
yields the joint and sum distributions for uniform random variables
$(U_1,U_2)$ satisfying $U = U_1 + U_2 = 1$:
%% almost surely: 
\begin{eqnarray*}
F_{U_1,U_2}(u_1,u_2) = \max\{u_1+u_2-1,0\} & & f_{U_1,U_2}(u_1,u_2) = \mathbf{1}_{u_1+u_2=1} \nonumber \\
F_U(u) = \mathbf{1}_{u \geq 1} & & 
f_U(u) = \mathbf{1}_{u = 1}. 
\label{eq:conaffnegcdfpdf}
\end{eqnarray*}
The case of perfect negative correlation is simpler to analyze than
the case of perfect positive correlation.  All users now see the same
utility of $1+ex-c$ for the bundle.  The following corollary of
Prop.\ \ref{prop:conaffbunprobadopt} shows that when $c<1$ all users
immediately adopt, while seeding to an adoption level of $x = c-1$ is
needed to ensure full adoption when $e>c-1$, and adoption is never
feasible when $e<c-1$.    
\begin{corollary}
\label{cor:conaffnegprobadopt}
The probability of bundle adoption $h(x)$ in
Prop.\ \ref{prop:conaffbunprobadopt} under aggregate affinity
$U=U_1+U_2$ formed from perfectly negatively correlated uniform
affinities $(U_1,U_2)$ satisfying $U_1 + U_2=1$ is 
\begin{equation}
h(x) = \left\{ \begin{array}{ll}
0, \; & x < m \\
1, \; & x \geq m \end{array} \right.  
\end{equation}
The two possible equilibria (eq.) are $x^* \in \{0,1\}$, with
conditions for each: 
\begin{equation*}
0 \mbox{ is stable eq.} \Leftrightarrow c > 1, \quad
1 \mbox{ is stable eq.} \Leftrightarrow e > c-1 \label{eq:conaffnegeq1}
\end{equation*}
The lowest stable equilibrium (lseq.) adoption level for each $(c,e)$ is
\begin{equation*}
0 \mbox{ is lseq.} \Leftrightarrow c > 1, \quad
1 \mbox{ is lseq.} \Leftrightarrow c < 1 \label{eq:conaffnegeq2}
\end{equation*}
\end{corollary}
We next compare the lowest stable equilibria without ($(x_1^*,x_2^*)$)
and with ($(x^*,x^*)$) bundling as a function of the system parameters
$(e_1,e_2,c_1,c_2)$.  In general, under perfect negative correlation,
the overall cost of the bundle is the dominant factor in determining
whether bundling is beneficial.  As shown, below, this yields very
different outcomes when compared to the case of perfect positive
correlation.  
{\small
\begin{equation*}
\begin{array}{c|cc||cc}
& & & 0 & 1 \\ \hline
& & & c > 1 & c < 1  \\ \hline \hline
% Row 1
(0,0) & c_1 > 1 & c_2 > 1 & SS & WW \\
& & & \mbox{True} & \mbox{False} \\ \hline
% Row 2
\left(0,\frac{1-c_2}{1-e_2} \right) & c_1 > 1 & e_2 < c_2 < 1 & SL & WW \\ 
& & & \mbox{True} & \mbox{False} \\ \hline 
% Row 3
(0,1) & c_1 > 1 & c_2 < e_2 \land 1 & SL & WS \\ 
& & & \mbox{True} & \mbox{False} \\ \hline 
% Row 4
\left(\frac{1-c_1}{1-e_1},0\right) & e_1 < c_1 < 1 & c_2 > 1 & LS & WW \\ 
& & & \mbox{True} & \mbox{False} \\ \hline
% Row 5 
(1,0) & c_1 < e_2 \land 1 & c_2 > 1 & LS & SW \\ 
& & & \mbox{True} & \mbox{False} \\ \hline
% Row 6
\left(\frac{1-c_1}{1-e_1},\frac{1-c_2}{1-e_2}\right) & e_1 < c_1 < 1 & e_2 < c_2 < 1 & LL & WW \\ \hline
% Row 7
\left(\frac{1-c_1}{1-e_1},1 \right) & e_1 < c_1 < 1 & c_2 < e_2 \land 1 & LL & WS \\ \hline
% Row 8 
\left(1,\frac{1-c_2}{1-e_2} \right) & c_1 < e_2 \land 1 & e_2 < c_2 < 1 & LL & SW \\ \hline
% Row 9
(1,1) & c_1 < e_1 \land 1 & c_2 < e_2 \land 1 & LL & SS \\ \hline
\end{array}
\end{equation*}
}
%%Several comments bear mention, which highlight differences with the 
%%case of perfect positive correlation.  
First, seven rather than eight of the nine equilibrium change pairs
$(\Delta_1,\Delta_2)$ are possible.  The two missing entries are $WL$
and $LW$ (as opposed to $LL$ for perfect positive correlation), \ie it
is not possible for the adoption levels of the two services to
simultaneously increase and decrease, respectively.  Second, if either
equilibrium under separate offerings is zero then the bundled
equilibrium is zero, \ie both services must be individually viable for
a bundled offering to succeed.  Again, this is unlike the perfect
positive correlation case, where pairing a service that was not viable
on its own with a more successful one, could result in a non-zero
adoption for the bundle (and even in some cases in a $WW$ outcome).
Third, when both equilibria under separate offerings are nonzero, the
bundled equilibria may be better than or equal to both equilibria, or
may be worse than or equal to both equilibria.  For example, the
separate offering equilibria pair $(x_1^*,x_2^*) =
((1-c_1)/(1-e_1),(1-c_2)/(1-e_2))$ (which requires $e_1 < c_1 < 1$ and
$e_2 < c_2 < 1$) may yield a bundled equilibria of $(0,0)$ if $c>1$ or
$(1,1)$ if $c<1$.  In the case of perfect positive correlation, the
bundle equilibrium is always at some intermediate value between the
two stand-alone equilibria.

The next section considers the intermediate configuration of
independent affinities in an attempt to explore when and how changes
occur between those two extremes.
%This corresponds to a very different type of outcomes than in the
%previous subsection, and in particular it can give rise to true
%win-win situations, \ie both services experience a strictly higher
%level of adoption than if offered alone.  This arises whenever $c<1$
%(the bundle experiences an adoption level of $x^*=1$), while both
%services verify $0\leq e_i<c_i<1$ for both $i \in \{1,2\}$ so that
%according to \eqref{eq:equil_s} their individual adoption levels are
%of the form $x_i^*=(1-c_i)/(1-e_i)<1, i=1,2$.  
\subsection{Independent affinities}
\label{ssec:uncor}

% correlation -> independent?  give CDF and PDF for U give h(x)
% characterize equilibria regions and stability find lowest stable
% equilibria compare bundled equilibria with separate offering
% equilibria discuss related work since most work on bundling with
% externalities assumes independent aff.

Specialization for $\rho = 0$ of Props.\ \ref{prop:jointcdfpdfcorrcont} and 
\ref{prop:conaffsumpdfcdf} in the Appendix yields the joint and sum
distributions for independent uniform random variables:
\begin{eqnarray*}
F_{U_1,U_2}(u_1,u_2) = u_1 u_2 & & f_{U_1,U_2}(u_1,u_2) = 1 \nonumber \\
F_U(u) = \left\{ \begin{array}{ll} \frac{u^2}{2}, \; & 0 \leq u \leq 1 \\
1-\frac{(2-u)^2}{2}, \; & 1 < u \leq 2 \end{array} \right. & & 
f_U(u) = \left\{ \begin{array}{ll} u, \; & 0 \leq u \leq 1 \\
2-u, \; & 1 < u \leq 2 \end{array} \right. . 
\label{eq:conaffindcdfpdf}
\end{eqnarray*}
This yields the following corollary to
Prop.\ \ref{prop:conaffbunprobadopt} of the Appendix.

\begin{corollary}
\label{cor:conaffindprobadopt}
%% in Prop.\ \ref{prop:conaffbunprobadopt} 
Under aggregate affinity $U$ formed from independent uniform
affinities $(U_1,U_2)$ with distribution $F_U(.)$, the probability of
bundle adoption $h(x)$ is:
%%in \Eqref{eq:conaffindcdfpdf} is
\begin{equation}
h(x) = \left\{ \begin{array}{lrcccl}
0, \; & & & x & \leq & l \\
\frac{1}{2}(2-(c-ex))^2, \; & l & < & x & \leq & m \\
\frac{1}{2}(2-(c-ex)^2), \; & m & < & x & \leq & r \\
1, \; & r & < & x & &  
\end{array} \right. 
\end{equation}
which is convex increasing on $l < x \leq m$ and concave increasing on
$m < x \leq r$ (recall $h(m) = 1/2$).  Besides $x^* \in \{0,1\}$, the
possible equilibria in $(0,1)$ are: 
\begin{equation}
x^* \in \left\{ \begin{array}{lll}
\xi^*_{l,\pm} \equiv \frac{1}{e^2}\left( (c-2)e+1 \pm \sqrt{2(c-2)e+1}\right), \; & l < x^* \leq m \\
\xi^*_{r,\pm} \equiv \frac{1}{e^2}\left( c e - 1 \pm \sqrt{2(e-c)e+1}\right),\; & m < x^* \leq r 
\end{array} \right.
\label{eq:conaffindeq1}
\end{equation}
The regions on the $(c,e)$ plane where these equilibria exist are
\begin{eqnarray*}
\Rmc_{l,\pm} &=& \{ (c,e) : \max\{l,0\} \leq \xi^*_{l,\pm} \leq \min\{m,1\}\} \nonumber \\
\Rmc_{r,\pm} &=& \{ (c,e) : \max\{m,0\} \leq \xi^*_{r,\pm} \leq \min\{r,1\}\} 
\label{eq:conaffindeqreg1}
\end{eqnarray*}
\end{corollary}

\begin{proof}
The first two derivatives of $h(x)$ are
\begin{equation*}
h'(x) = \left\{ \begin{array}{lrcccl}
e(2 - (c - ex)),\; & l & < & x & \leq & m \\
e(c-ex),\; & m & < & x  & \leq & r 
\end{array} \right. , ~~~ 
h^{''}(x) = \left\{ \begin{array}{lrcccl}
e^2, \; & l & < & x & \leq & m \\
-e^2, \; & m & < & x  & \leq & r 
\end{array} \right.
\end{equation*}
The equilibria are the solutions of $h(x) = x$, i.e., 
\begin{eqnarray*}
\frac{1}{2}(2 - (c-ex))^2 = x & \Leftrightarrow & e^2 x^2 - 2((c-2)e+1)x + (c-2)^2 = 0, \; l < x \leq m \nonumber \\
\frac{1}{2}(2-(c-ex)^2) = x & \Leftrightarrow & e^2 x^2 - 2 (ce-1)x + (c^2-2) = 0, \; m < x \leq r
\end{eqnarray*}
The solutions are given by \Eqref{eq:conaffindeq1}.  
\end{proof}

%% The regions in \eqref{eq:conaffindeqreg1} are illustrated in
%% Fig.\ \ref{fig:conaffindeqregcomb1}.  Observe $i)$ $e = 2(c-1)$ is the
%% solution of $\xi_{l,\pm}^* = m$ and $\xi_{r,\pm}^* = m$, $ii)$ $e =
%% 1/(2(2-c))$ is the solution of $2(c-2)e+1=0$ where $2(c-2)e+1$ is the
%% discriminant of $\xi_{l,\pm}^*$ in \eqref{eq:conaffindeq1}, and $iii)$
%% $e = \frac{1}{2}(c + \sqrt{c^2-2})$ is the solution of $2(e-c)e+1=0$
%% where $2(e-c)e+1$ is the discriminant of $\xi_{r,\pm}^*$ in
%% \eqref{eq:conaffindeq1}.   

An explicit comparison of the equilibria with and without bundling as
in the two previous sections appears to be complicated.  Without
bundling, the equilibria $(x_1^*,x_2^*)$ are such that $x_i^*$ depends
upon $(c_i,e_i)$ as in Prop.\ \ref{prop:conaffsepoffequ}.  With
bundling, the equilibria $(x^*,x^*)$ is such that $x^*$ depends upon
$(c,e)$ (where $c = c_1+c_2$ and $e=e_1+e_2$) as in
Cor.\ \ref{cor:conaffindprobadopt}.  \fig{fig:conaffindeqregcomb1} in
the Appendix illustrates the complex shapes of the bundled equilibria
regions even in this relatively simple case of independent affinities.

%% \begin{figure}[!ht]
%% \centering
%% \includegraphics[width=\textwidth]{../figures/FigConAffIndEqRegComb1}
%% \caption{The four equilibria regions $\Rmc_{l,\pm},\Rmc_{r,\pm}$ of
%%   the $(c,e)$ plane in \eqref{eq:conaffindeqreg1} for continuous and
%%   independent affinities $(U_1,U_2)$.} 
%% \label{fig:conaffindeqregcomb1}
%% \end{figure}

\subsection{Summary}
\label{ssec:conaffsum}

Sections~\ref{ssec:perposcor} and~\ref{ssec:pernegcor} hint at a
transition in the impact of correlation on bundling.
Section~\ref{ssec:uncor} unfortunately illustrates that while a direct
analysis is feasible, it is cumbersome, which makes extracting insight
into when bundling can improve adoption challenging.
%%.  More importantly, even for special cases, \eg independent
%%valuations, 
As a result, the next section introduces a discrete affinity model
that preserves users' heterogeneity, but allows us to explicitly
explore the impact of correlation.  Section~\ref{sec:impact_rho_con}
assesses through numerical investigations the robustness of the
results obtained using this simplified discrete model.

% SW -- Outline -- Summary of results - really want to explore
% transition between the two, hard to do even for "middle" case of
% independent (perhaps include picture of regions for separate
% offering equilibria).  This is a transition to motivate the discrete
% affinity model. 

\section{Discrete affinities}
\label{sec:disaff}

In this section, we model user affinities as a pair of Bernoulli
random variables $(U_1,U_2) \in \{0,1\}^2$ with joint distribution
parameterized by $p \in [0,1]$:  
\begin{equation*}
\label{eq:disaffjointdisbn}
\begin{array}{r|cc|c}
U_1 \backslash U_2 & 0 & 1 & \\ \hline
0 & (1-p)/2 & p/2 & 1/2 \\
1 & p/2 & (1-p)/2 & 1/2 \\ \hline
& 1/2 & 1/2 & 
\end{array}
\end{equation*}
The user population consists of four types: negative affinities for
both services $(0,0)$, positive affinities for both services $(1,1)$,
and mixed service affinities $(0,1)$ and $(1,0)$. Note the marginals
are independent of the parameter $p$, and are in fact uniform, \ie
$\Pbb(U_1 = 1) = \Pbb(U_2 = 1) = 1/2$.  Thus, exactly half of the
population has a positive affinity for each service, regardless of
$p$.  Although the discrete model is a simplification of the
continuous model of Section~\ref{sec:conaff}, it facilitates study of
the impact of correlation in user service affinities.  The correlation
between $(U_1,U_2)$ is
\begin{equation*}
\rho = \frac{\Ebb[U_1 U_2] - \Ebb[U_1] \Ebb[U_2]}{\sqrt{\mathrm{Var}(U_1)\mathrm{Var}(U_2)}} = \frac{\frac{1-p}{2} - \frac{1}{2} \times \frac{1}{2}}{\sqrt{\frac{1}{4} \times \frac{1}{4}}} = 1 - 2 p, 
\end{equation*}
which ranges from $\rho = -1$ for $p = 1$ (all users have mixed
affinities, $p_{01} = p_{10} = 1/2$) up to $\rho = +1$ for $p = 0$
(all users' affinities are either both positive or both negative,
$p_{00} = p_{11} = 1/2$).  

\subsection{Separate offerings}
\label{ssec:disaffsep}

The probability of a user adopting service $i \in \{1,2\}$ under separate service offerings and the resulting equilibria are given in the following proposition. 
\begin{proposition}
\label{pro:disaffsep}
When $U_i$ is uniformly distributed on $\{0,1\}$, the probability of user adoption of service $i$ in \Eqref{eq:equil_s} becomes
\begin{equation}
h_i(x_i) = 
\left\{ \begin{array}{lrcccl}
0, \; & & & x_i & \leq & l_i \\
\frac{1}{2}, \; & l_i & < & x_i & \leq & r_i \\
1, \; & r_i & < & x_i & & 
\end{array} \right.
\end{equation}
for adoption thresholds $l_i \equiv \frac{c_i-1}{e_i}$ and $r_i \equiv \frac{c_i}{e_i}$.  The three possible equilibria are $x_i^* \in \{0,1/2,1\}$.  The conditions for each equilibrium (eq.) are:
\begin{equation}
\begin{array}{ccrcccl}
0 \mbox{ is stable eq.} & \Leftrightarrow & 1 & \leq & c_i & &  \nonumber \\
\frac{1}{2} \mbox{ is stable eq.} & \Leftrightarrow & 2(c_i-1) & \leq & e_i & \leq & 2 c_i \nonumber \\
1 \mbox{ is stable eq.} & \Leftrightarrow & c_i & \leq & e_i & & 
\end{array}
\end{equation}
The lowest stable equilibrium (lseq.) adoption level for each $(c_i,e_i)$ is
\begin{equation}
\begin{array}{cccrclcrcl}
0 & \mbox{ is lseq.} & \Leftrightarrow & c_i & \geq & 1 &&&& \nonumber \\
\frac{1}{2} & \mbox{ is lseq.} & \Leftrightarrow & c_i & \leq & 1 & \mbox{ and } & e_i &  \leq & 2 c_i \nonumber \\ 
1 & \mbox{ is lseq.} & \Leftrightarrow & c_i  & \leq & 1 & \mbox{ and } & e_i & \geq & 2 c_i  
\end{array}
\end{equation}
\end{proposition}

The proof of Prop. \ref{pro:disaffsep} is straightforward and is
omitted. 
%% As illustrated in Fig.\ \ref{fig:disaffsep}, 
All seven nonempty subsets of $\{0,1/2,1\}$ may coexist as equilibria,
and all equilibria are stable.  If costs are high $(c_i \geq 1)$ then
no adoption is possible; likewise if the externality is high $(e_i
\geq c_i$) then full adoption is possible. Intermediate-level ($x_i^*
= 1/2$) adoption is possible for externalities that are moderate with
respect to the cost.     

%% \begin{figure}[!ht]
%% \centering
%% \includegraphics[width=\textwidth]{../figures/FigDisAffSep1}
%% \caption{
%% Illustration of Prop.\ \ref{pro:disaffsep}.  {\bf Left:} the eight
%% possible orderings of $\{l_i,r_i\}$ with $\{0,1/2,1\}$ each
%% determine the subset of $\{0,1/2,1\}$ that are equilibria. All
%% equilibria are stable.  {\bf Right:} the $(c_i,e_i)$ plane and the
%% equilibria in each region.  {\bf Bottom right:} partition of the
%% $(c_i,e_i)$ plane according to lowest stable equilibria.}   
%% \label{fig:disaffsep}
%% \end{figure}

\subsection{Bundled offerings}
\label{ssec:disaffbun}

The probability of a user adopting a bundled service offering and the
resulting equilibria are given in the following proposition. 

\begin{proposition}
\label{pro:disaffbun}
When $(U_1,U_2)$ are distributed on $\{0,1\}^2$ according to \Eqref{eq:disaffjointdisbn} with parameter $p \in [0,1]$ and correlation $\rho = 1-2p \in [-1,1]$, the probability of user adoption of the bundle in \Eqref{eq:BunCCDF} becomes
\begin{equation}
\label{eq:disaffbunprobadopt}
h(x) = \left\{ \begin{array}{lrcccl}
0, \; & & & x & \leq & l \\
\frac{1+\rho}{4}, \; & l & < & x & \leq & m \\
\frac{3-\rho}{4}, \; & m & < & x & \leq & r \\
%\frac{1-p}{2} = \frac{1+\rho}{4}, \; & l < x \leq m \\
%\frac{1+p}{2} = \frac{3-\rho}{4}, \; & m < x \leq r \\
1, \; & r & < & x & & 
\end{array} \right., 
\end{equation}
for adoption thresholds $l \equiv \frac{c-2}{e}$, $m \equiv \frac{c-1}{e}$, $r \equiv \frac{c}{e}$.  The four possible equilibria are 
%$x^* \in \{0,(1-p)/2,(1+p)/2,1\}$, or equivalently, 
$x^* \in \{0,(1+\rho)/4,(3-\rho)/4,1\}$.  The conditions for each equilibrium (eq.) are:
\begin{equation}
\begin{array}{cccccrcccl}
0 & \mbox{ is stable eq.} & \Leftrightarrow & c & \geq & 2 & & \\
\frac{1+\rho}{4} & \mbox{ is stable eq.} & \Leftrightarrow & \frac{4(c-2)}{1+\rho} & \leq & e & \leq & \frac{4(c-1)}{1+\rho} \\ 
\frac{3-\rho}{4} & \mbox{ is stable eq.} & \Leftrightarrow & \frac{4(c-1)}{3-\rho} & \leq & e & \leq & \frac{4c}{3-\rho} \\ 
1 & \mbox{ is stable eq.} & \Leftrightarrow & c & \leq & e & & \\
\end{array}
\end{equation}
The lowest stable equilibrium (lseq.) adoption level for each $(c_i,e_i)$ is
\begin{equation}
\label{eq:disaffbunofflseq}
\begin{array}{cccccrcccl}
0 & \mbox{ is lseq.} & \Leftrightarrow & c \geq 2 & & & & & & \\
\frac{1+\rho}{4} & \mbox{ is lseq.} & \Leftrightarrow & c < 2 & \mbox{and} & 0 & \leq & e & \leq & \frac{4(c-1)}{1+\rho} \\ 
\frac{3-\rho}{4} & \mbox{ is lseq.} & \Leftrightarrow & c < 2 & \mbox{and} & \frac{4(c-1)}{1+\rho} & < & e & \leq & \frac{4c}{3-\rho} \\ 
1 & \mbox{ is lseq.} & \Leftrightarrow & c < 2 & \mbox{and} & & & e & > & \max \left\{\frac{4c}{3-\rho},\frac{4(c-1)}{1+\rho} \right\}
\end{array}
\end{equation}
\end{proposition}

\begin{proof}
Conditioning on $(U_1,U_2)$ gives the user adoption probability as:  
\begin{eqnarray}
h(x) 
&=& \Pbb(V(x) > 0) = \Pbb(U > c - e x) \nonumber \\
&=& \Pbb(U > c - e x | (U_1,U_2) = (0,0)) \Pbb((U_1,U_2) = (0,0)) + \nonumber \\
& & \Pbb(U > c - e x | (U_1,U_2) = (0,1)) \Pbb((U_1,U_2) = (0,1)) + \nonumber \\ 
& & \Pbb(U > c - e x | (U_1,U_2) = (1,0)) \Pbb((U_1,U_2) = (1,0)) + \nonumber \\
& & \Pbb(U > c - e x | (U_1,U_2) = (1,1)) \Pbb((U_1,U_2) = (1,1)) \nonumber \\ 
&=& \frac{1-p}{2} \Pbb(e x > c) + \frac{p}{2} \Pbb(1 + e x > c) + \frac{p}{2} \Pbb(1 + e x > c) + \frac{1-p}{2} \Pbb(2 + e x > c) \nonumber \\ 
&=& \frac{1-p}{2} \mathbf{1}_{x > l} + p \mathbf{1}_{x > m} +
\frac{1-p}{2} \mathbf{1}_{x > r} 
\end{eqnarray}
The characterization of the equilibria and the lseq.\ are straightforward and are omitted.
\end{proof}

As with separate offerings,
no adoption is possible if costs are high $(c \geq 2)$, and full
adoption is possible if the externality is high $(e \geq c)$. The
intermediate equilibria $(1+\rho)/4,(3-\rho)/4$ are possible when the
externality is moderate with respect to the cost.  Of interest is
identifying regions where bundling yields a higher adoption
equilibirium, \ie $WW$ scenarios, and and in particular how this
outcome may be affected by $\rho$.
%% Of interest is comparing the bottom-right plots of
%% Fig.\ \ref{fig:disaffsep} and Fig.\ \ref{fig:disaffbun}, to
%% identify the regions where bundling yields a higher adoption
%% equilibrium\footnote{Note though that the bottom-left plot of
%% Fig.\ \ref{fig:disaffbun} is for the specific value of $\rho=0$.
%% Different regions, and therefore outcomes arise as $\rho$ varies.
%% %% Characterizing $\rho$'s ability to affect the outcome is one of
%% our %% goals.  }.
Exploring this issue is the topic of Section~\ref{ssec:disaffsum}.

%% \begin{figure}[!ht]
%% \centering
%% \includegraphics[width=\textwidth]{../figures/FigDisAffBun1}
%% \caption{
%% Illustration of Prop.\ \ref{pro:disaffbun}.  {\bf Top:} the
%% probability of adoption $h(x)$ in \eqref{eq:disaffbunprobadopt} in
%% terms of $p$ (left) and $\rho$ (right).  {\bf Bottom:} the regions
%% of the $(c,e)$ plane for each of the four equilibria
%% $\{0,(1+\rho)/4,(3-\rho)/4,1\}$.  Also shown are the superimposed
%% boundaries of the four equilibria regions, as well as the partition
%% of the $(c,e)$ plane according to the lowest stable equilibria.
%% The figures are shown for $\rho=0$.} 
%% \label{fig:disaffbun}
%% \end{figure}

\subsection{Bundling's impact on equilibria}
\label{ssec:disaffsum}
\begin{figure*}[htbp] 
\begin{tiny}
\begin{equation*}
\begin{array}{c|cc||cccc}
& & \mbox{BunEq}\Rightarrow & 0 & \frac{1+\rho}{4} & \frac{3-\rho}{4} & 1 \\ \hline
& & & c > 2 & c < 2 & c < 2  & c < 2  \\ 
& & & & (1+\rho) e < 4(c-1) & (1+\rho) e > 4(c-1)  & 
(1+\rho) e > 4(c-1) \\
\mbox{SepEq} &\mbox{SepEq conditions} & & & & (3-\rho) e < 4 c &
(3-\rho) e > 4 c  \\ \hline \hline  

(0,0) & c_1 > 1 & c_2 > 1 & SS & WW & WW & WW \\
& & & \mbox{True} & \mbox{False} & \mbox{False} & \mbox{False} \\ \hline

(0,1/2) & c_1 > 1 & c_2 < 1 &  SL & WL & WW & WW \\
& & e_2 < 2 c_2 &  &  &  &  \\ \hline 

(0,1) & c_1 > 1 & c_2 < 1 & SL & WL & WL & WS \\ 
& & e_2 > 2 c_2 & & & \\ \hline 

(1/2,0) & c_1 < 1 & c_2 > 1 & LS & LW & WW & WW \\
& e_1 < 2 c_1 & & & & \\ \hline 

(1,0) & c_1 < 1 & c_2 > 1 & LS & LW & LW & SW \\
& e_1 > 2 c_1 & & & & \\ \hline 

(1/2,1/2) & c_1 < 1 & c_2 < 1 & LL & LL & WW & WW \\
& e_1 < 2 c_1 & e_2 < 2 c_2 & \mbox{False} &   & & \\ \hline

(1/2,1) & c_1 < 1 & c_2 < 1 & LL & LL & WL & WS \\ 
& e_1 < 2 c_1 & e_2 > 2 c_2 & \mbox{False} & & & \\ \hline

(1,1/2) & c_1 < 1 & c_2 < 1 & LL & LL & LW & SW \\
& e_1 > 2 c_1 & e_2 < 2 c_1 & \mbox{False} & & & \\ \hline

(1,1) & c_1 < 1 & c_2 < 1 & LL & LL & LL & SS \\ 
& e_1 > 2 c_1 & e_2 > 2 c_1 & \mbox{False} &  &  &  \\ \hline 
\end{array}
\end{equation*}
\end{tiny}
\caption{Rows: equilibria under separate offerings; Columns:
  equilibria under bundling.\\  
Individual table entries show changes in equilibrium under bundling
  (Same, Win, Lose) for each service and whether they can occur
  (True/False).
\label{fig:disafftab}}
\end{figure*}
There are $3\times 3\times 4 = 36$ possible lseq.~combinations
$((x_1^*,x_2^*),x^*)$ where $(x_1^*,x_2^*)$ is the separate offering
lowest equilibria, and $x^*$ is the bundling lowest equilibrium.  The
table of \fig{fig:disafftab} lists all $36$ combinations and
identifies the conditions under which each holds and whether
bundling is beneficial or not.
%% , \ie increases equilibrium adoption
%% levels compared to separate offerings
As in Section~\ref{sec:conaff},
equilibria under separate offerings form the rows, while the four
equilibria under bundling form the (left-most) columns.  The row
headings (second column) give the requirements on $(c_1,c_2,e_1,e_2)$
for a particular pair of separate offering lower equilibria.  The
column headings (second row) give the requirements on $(c,e,\rho)$ for a
particular bundled equilibrium to be the lowest equilibrium.

Individual entries in the table identify how bundled equilibria
compare to equilibria under separate offerings, \ie as before a
``win'' (W), a ``Loss'' (L), or the ``Same'' (S), and whether
individual combinations are feasible (True) or not (False).  Note that
in several instances, row and column conditions are redundant, \eg
$c_1>1$ and $c_2>1$ obviously imply $c>2$, so that
simplifications are possible.
%% in the constraints associated with some entries.  
For clarity of presentation, we omit
specifying those more compact requirements in the table.
%% Within each cell of
%% the table we use the following notations (where BEQ is bundled
%% equilibria and SEQ is separate offering equilibria):  
%% \begin{enumerate}
%% \item True: the BEQ will always take that value for the SEQ.
%% \item False: the BEQ will never take that value for the SEQ.
%% \item $*$: the BEQ will take that value under the given SEQ
%% provided the column heading conditions are met. 
%% % (some column conditions may be implied by row conditions). 
%% \item 'condition' $\cap$ 'condition': the BEQ will take that value
%% under the given SEQ provided the specified conditions are true. 
%% \item $* \cap \mbox{condition}$: the BEQ will take that value under
%% the given SEQ provided the column heading conditions are met, and
%% the condition is true.   
%% \end{enumerate}
%% Each cell is characterized in terms of the relative magnitude of
%% bundled versus separate equilibria, \ie a ``win'' (W), 
%% %denotes an increased adoption level, 
%% a ``loss'' (L),
%% % denotes a decreased adoption level, 
%% and ``same'' (S).
% denotes the adoption level is unaffected.  The two letters denote
% this change in adoption level for each of the two services.  
% All nine combinations occur in the table.  
%% Fig.\ \ref{fig:disaffexa} illustrates the relationship among the
%% separate and bundled equilibria. 

Several observations follow from the table,
%% of \fig{fig:disafftab}, 
and in particular how $\rho$ affects the
emergence of $WW$ combinations. Of note is that the
configurations that yield $WW$ outcomes are qualitatively consistent
with those of Section~\ref{ssec:perposcor}, \eg combining a low-cost,
low externality technology, with a high-cost, high
externality one can improve adoption for both.  The
table, however, also reveals a more ambivalent role for
%% the affinity correlation parameter 
$\rho$ than the two extreme configurations of
Sections~\ref{ssec:perposcor} and~\ref{ssec:pernegcor} seemed to
indicate.  In particular, consider the conditions $(1+\rho)e>4(c-1)$
and $(3-\rho)e>4c$ that are required to hold for $1$ to be an
equilibrium under bundling.  Increasing (decreasing) $\rho$ makes it
easier (harder) for the
first condition to be met, but is clearly detrimental (beneficial) to
the second. 
%% Conversely, decreasing $\rho$ has the opposite effect.  
Similarly, varying $\rho$ can also allow the emergence of $WW$
scenarios present in Section~\ref{ssec:pernegcor} but
not~\ref{ssec:perposcor}, \ie combining two middling technologies,
$(x_1^*,x_2^*)=(1/2,1/2)$, can benefit both under certain conditions.
The next section investigates this more extensively.

%% \fig{fig:disaffexa} offers a representative illustration.  
%% \emph{RG: Add discussion based on the results of what we can observe
%%   from the figure once we have it.}

%% Specifically, the figure plots the regions of separate and bundled
%% equilibria as a function of $c_1$ ($x$-axis) and $c_2$ ($y$-axis),
%% for $e_1=e_2=1$, and for three different values of $\rho$, namely,
%% $\rho=-1/3, 0, 1/3$ (from left to right).  In all three cases, the
%% ``WW'' region is the shaded triangular-shaped region shown in the
%% figures.

Other observations are also possible from the table, and we summarize
next some of the more relevant ones.  First, if $(0,0)$ is the
separate offering equilibrium then $0$ is the bundled equilibrium, \ie
bundling cannot help.  This is because $c_1 > 1$ and $c_2 > 1$ implies
$c > 2$.  Second, if $(1,1)$ is the separate offering equilibrium,
then it is possible for the bundled equilibrium to be either
$(1+\rho)/4$ or $(3-\rho)/4$, \ie bundling can result in an $LL$
outcome. Third, if the separate offering equilibria are both non-zero,
then bundling cannot cause the equilibrium to drop to zero, but it can
cause it to drop (to $(1+\rho)/4$, which can be made arbitrarily close
to $0$).  This happens when $(1+\rho)e < 4 (c-1)$, \ie when the bundle
cost is relatively large $(c>1)$ and when the correlation coefficient
$\rho$ is small enough.  Fourth, if the separate offering equilibria
are both below one but at least one is non-zero, then bundling can
increase the equilibrium to $(3-\rho)/4$ or $1$, provided the bundle's
cost is not too high $(c<2)$.  For example, when the separate
offerings equilibria are $(0,1/2)$, the bundled offering equilibrium
is either $(3-\rho)/4$ or $1$ provided $c<2$ and $(1+\rho)e > 4(c-1)$.
In the next section, we explore further the impact that $\rho$ has on
the potential benefits that bundling can yield.

%% Note in general the ambiguous role of $\rho$ that can affect not
%% only the value of the bundled adoption equilibrium, but also which
%% (lseq) equilibrium is realized.  For instance, in the last example
%% of separate offerings equilibria of $(1/2,1/2)$, decreasing $\rho$
%% increases the value of the bundled equilibrium $(3-\rho)/4$ (from
%% $1/2$ for $\rho=1$ to $1$ for $\rho=-1$), while at the same time
%% making it more difficult for the inequality $(1+\rho)e > 4(c-1)$ to
%% hold. 
%% meeting the required inequalities in this and other scenarios. 

%% \begin{figure}[!ht]
%% \centering
%% \includegraphics[width=\textwidth]{../figures/FigDisAffExa}
%% \caption{Illustration of the relationship between the separate (red)
%%   and bundled (blue) equilibria for discrete affinities for $e_1 = e_2
%%   = 1$ and $\rho \in \{-1/3,0,1/3\}$.}   
%% \label{fig:disaffexa}
%% \vspace{-0.45cm}
%% \end{figure}

\section{Guidelines and interpretations}
\label{sec:guidelines}

The traditional ``wisdom'' in developing bundling strategies, \eg
see~\cite{VenMah2009}, is that bundling is typically most effective in
the presence of \emph{negative correlation} in user valuations
(reservation prices).  The intuition is that bundling reduces
heterogeneity in users' valuations, which facilitates the selection of
a ``price'' for a bundled offering that results in an overall higher
profit (see~\cite[Section 2.3]{VenMah2009}).

%% This is easily illustrated with a simple two products, $X$ and $Y$,
%% and two customers, $A$ and $B$, example.  Assume that $A$ is
%% willing to pay $p_1$ and $p_2<p_1$ for products $X$ and $Y$,
%% respectively, while $B$'s willingness to pay for $X$ and $Y$ is
%% $p_2, p_1$, respectively.  In other words, $A$ and $B$'s valuations
%% for the two products are perfectly negatively correlated.  For the
%% sake of illustration, assume $p_1=5$ and $p_2=3$.  It is then easy
%% to see that the optimal prices $p_1^*$ and $p_2^*$ when the two
%% products are offered separately are both equal to $3$ for a total
%% profit (assuming zero marginal costs) of $12$.  In contrast the
%% bundle's optimal price is $p^*=8$ for a total profit of $16>12$.
%% In contrast, if the two users' willingness to pay had been
%% perfectly positively correlated, then bundling yields no benefit
%% over separate offerings\footnote{In general it is possible for
%% bundling to lower profit by preventing users who may have bought
%% individual products from purchasing the bundle.}.

There are obviously differences between the profit maximization goal
of traditional bundling strategies, and our goal of maximizing
adoption given a fixed adoption cost that will typically be different
from the price that would optimize profit\footnote{This is not to say
  that it is not if interest to explore how changes in cost, \eg
  through incentives, affect adoption, but this aspect is beyond the
  focus of this initial investigation.}.  The other important
difference between our formulation and that of traditional bundling
strategies is the presence of externalities.  Hence, we can expect
both factors to contribute to possible differences in outcomes, with
the latter (presence of externalities) likely to have a stronger
influence.

In particular, it is relatively easy to see from \Eqrefs{eq:a}{eq:b}
that without externalities, assessing whether bundling benefits
adoption is straightforward.  Specifically, when services are
offered separately, adoption levels are simply equal to $1-F_1(c_1)$
and $1-F_2(c_2)$, where $F_i(x)$ is the CDF of users'
valuation for service~$i$.  Conversely, the adoption level of the
bundle is given by $1-F(c_1+c_2)$, where $F(x)$ is the CDF of the
random variable $U=U_1+U_2$ that captures the cumulative valuation of
the two services to a (random) user).  Hence, in the absence of
externalities, whether bundling is beneficial (improves overall
adoption) or not is solely a function of how the bundle's cost
compares to the cost of individual services.

On the other hand, as the models of both Sections~\ref{sec:conaff}
and~\ref{sec:disaff} revealed, more complex behaviors emerge when
externalities are present.  In particular, 
%% our initial results indicate that \emph{bundling is effective in
%% improving network technology adoption, and in particular create
%% win-win outcomes, when the technologies being bundled have {\bf
%% positively correlated} valuations, although too strong a
%% correlation will often mitigate this benefit.} 
the models revealed that $WW$ outcomes can arise under two general
scenarios.  The first involves bundling a service with a high
externality factor and a high adoption cost, with a second service
that enjoys middling cost and externality factor.  Alternatively, $WW$
outcomes may also arise from bundling two middling services that
alone cannot create sufficient externality value to reach a high level
of adoption, but which together could.  In both cases, correlation
$(\rho)$ in how individual users value the services can affect the
outcome.

\subsection{On the role of correlation (discrete model)}
\label{sec:impact_rho_dis}

\begin{figure}[htbp]
\centering
\hspace{-0.4cm}
\includegraphics[width=5.6in]{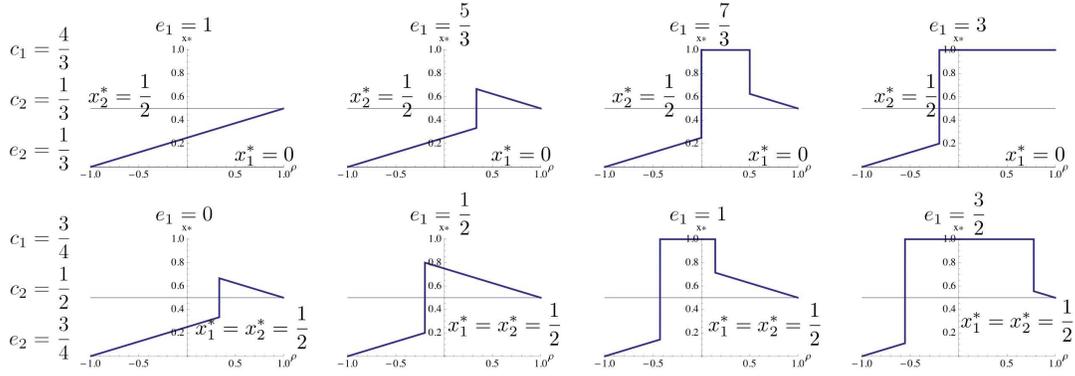}
\caption[caption]{\label{fig:bundling-disc}Impact of value
  correlation $(\rho)$ on bundle adoption 
  $(x^*)$ for different technology combinations (discrete model).
%% Technology~$1$: $c_1=4/3$ and (left to right) $e_1=1,3/2,2,5/2,3,$ so
%% that $x_1^*=0$.\\%\hspace{\textwidth}
%% Technology~$2$: $c_2=1/3, e_2=1/3$ so that $x_2^*=1/2$.} 
%% %\end{figure}
%% %\begin{figure}[htbp]
%% \centering
%% \hspace{-0.4cm}
%% \includegraphics[width=5.6in]{../figures/Bundling-discrete-1+1=3-function_of_rho(separate).pdf}
%% \caption[caption]{\label{fig:bundling-disc2}Impact of value
%%   correlation $(\rho)$ on bundle adoption 
%%   $(x^*)$ for different technology combinations.\\%\hspace{\textwidth}
%% Technology~$1$: $c_1=3/4$ and (left to right) $e_1=0,1/2,1,3/2$ so
%% that $x_1^*=1/2$.\\%\hspace{\textwidth}
%% Technology~$2$: $c_2=1/2, e_2=3/4$ so that $x_2^*=1/2$.
} 
\end{figure}

%% c1=3/4, c2=1/2, e2=3/4, and e1 = (0,1/2,1,3/2).  Note e1 <= 2 c1 =
%% 3/2 to have (x_1^*,x_2^*) = (1/2,1/2).   

The impact of $\rho$ is illustrated in \fig{fig:bundling-disc} that
plots as a function of $\rho\in[-1,1]$, the adoption level of a
technology bundle for different instances of the two above scenarios
under the discrete correlation model of Section~\ref{sec:disaff}.

Specifically, the upper part of \fig{fig:bundling-disc} displays
adoption levels when bundling two heterogeneous technologies.
Technology~$1$ has a high cost, $c_1=4/3,$ which prevents it from
taking off on its own, \ie its stand-alone adoption remains at
$x_1^*=0$, irrespective of its externality factor $e_1$.
Technology~$2$ has a low cost, $c_2=1/3,$ but marginal externality,
$e_2=1/3,$ so that $x_2^*=1/2$.  Combining the two technologies can
benefit both, but only when the externality $e_1$ of technology~$1$ is
high enough, \ie $e_1\geq 5/3$ (three right most plots).  When $e_1$
is low, \ie $e_1\leq 1$, technology~$1$ still benefits from being
bundled with technology~$2$, but the reverse is not true $(x^*\leq
1/2)$.  More interesting though than the impact of $e_1$ in creating a
$WW$ outcome, is the role of $\rho$.

Specifically, when $e_1$ is large enough, the benefits of bundling
arise only once $\rho$ exceeds a certain threshold.
This is because early adopters of the bundle are
driven primarily by the second technology, and under highly negative
correlation in technology valuations, the first technology contributes
added cost but little or no added value to those early adopters.
Hence, adoption stops quickly at a level below that of the second
technology offered alone.  As correlation increases, the number of
early adopters that derive positive utility from adopting the
bundle increases to a point where adoption can reach enough of a
critical mass to allow the externality effect of the first technology
to become effective.  This allows adoption to increase beyond what the
second technology alone would have realized.

Note though that further increases in correlation need not yield
additional improvements.  As a matter of fact, increasing $\rho$
beyond the threshold can lower adoption (second plot from the left,
$e_1=5/3$).  This is because as correlation increases, the potential
adoption ``base'' of the bundle narrows (both technologies appeal to
an increasingly similar set of users), which limits the adoption
equilibrium that can be reached.  This effect persists until the
externality factor of technology~$1$ is strong enough to allow the
bundle to reach full adoption (third and fourth plots from the left
for $e_1=7/3, 3$).  As the externality factor of technology~$1$
continues increasing, its strength becomes sufficient to preserve full
adoption for some range of $\rho$ beyond the initial threshold.
Further increases of $\rho$ outside that range can, however, result in
the adoption level of the bundle dropping again (third plot from the
left, $e_1=7/3$).  This is only avoided once the externality factor of
the first technology is strong enough that the range of $\rho$ values
for which no decline in bundle's adoption occurs extends all the way
to $\rho=1$ (right-most plot for $e_1=3$).

Conversely, the lower part of \fig{fig:bundling-disc} considers the
bundling of two ``middling'' technologies, which alone only realize a
relatively low adoption level $x_1^*=x_2^*=1/2$.  They both have
reasonably low costs, $c_1=3/4, c_2=1/2$, and can benefit from
bundling when their combined externality factor, $e=e_1+e_2$ is high
enough.  The four plots display (left to right) adoption as a function
of $\rho$ and for increasing values of $e$ ($e_2=3/4$ and $e_1$ varies
from $0$ to $3/2$).  They offer a qualitatively similar behavior as
the upper part of \fig{fig:bundling-disc}, albeit with a more limited
range, \eg $WL$ outcomes can be eliminated (if $\rho$ is high enough) and
decreases in adoption as $\rho$ keeps increasing cannot be avoided. 
This is not unexpected since the constraint that $x_1^*=x_2^*=1/2$
limits the range of costs and externality factors permissible.  

In the next section, we explore the extent to which the above
conclusions remain qualitatively valid under the more general model of
continuous affinities of Section~\ref{sec:conaff}.
\begin{figure}[htbp]
\centering
\hspace{-0.4cm}
\includegraphics[width=5.6in]{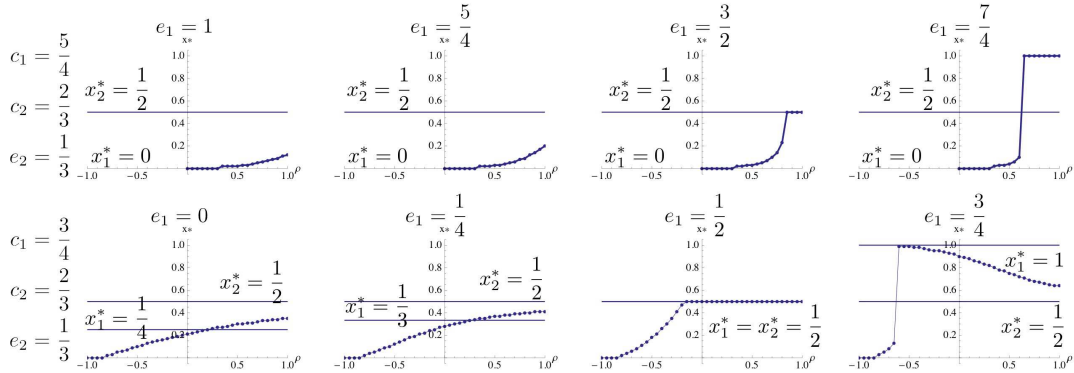}
\caption[caption]{\label{fig:bundling-cont}Impact of value
  correlation $(\rho)$ on bundle adoption 
  $(x^*)$ for different technology combinations (continuous model).}
\end{figure}

\subsection{On the role of correlation (continuous model)}
\label{sec:impact_rho_con}

The discrete affinity model of Section~\ref{sec:disaff} let us
explicitly account for the impact of correlation when bundling
services.  Its relative simplicity, however, raises the question of
whether the findings hold under more general (realistic) assumptions.
An exhaustive assessment is clearly impractical, and we limit
ourselves to the uniform distribution of Section~\ref{sec:conaff} to
offer initial evidence of the ``robustness'' of the results.  Because,
as mentioned in Section~\ref{ssec:uncor}, an analytical investigation
of uniformly distributed affinities under general correlation is
complex, we resort to a numerical approach.  Specifically, we consider
a pair of bundling scenarios similar to those of
\fig{fig:bundling-disc}, and numerically evaluate the bundle's 
adoption for different values of $\rho$.  The results are reported in
\fig{fig:bundling-cont}, which largely mirrors \fig{fig:bundling-disc}
with some differences as we briefly review.

The two sets of plots in \fig{fig:bundling-cont} clearly display the
presence of a threshold effect, where correlation $(\rho)$ needs to
exceed a certain minimum value before bundling becomes beneficial.
This is particularly so when combining two heterogeneous services; a
high-cost, high-externality one with a low-cost, low-externality one
(top set of plots).  Unlike the corresponding scenario in
\fig{fig:bundling-disc}, the jump in the bundle's adoption that occurs
after crossing the threshold is not followed by a decline in adoption
as $\rho$ further increases.  This is likely because under a uniform
distribution, the relative value of the externality after crossing the
threshold is sufficient to prevent declines in adoption for larger
values of $\rho$, \ie a scenario similar to that of the top right-most
plot of \fig{fig:bundling-disc}.  The potentially negative impact of
further increases in $\rho$ (beyond the threshold) is, however, seen in
the lower set of plots of \fig{fig:bundling-cont}.  In particular, the
right-most plot clearly displays that while $\rho$ needs to exceed a
threshold value of about $-0.4$ for the bundle to jump to full
adoption $(x^*=1)$, increasing $\rho$ beyond this value
results in progressively lower adoptions levels.

We note that the last scenario is an instance of a $SW$ rather than a
true $WW$ scenario, and under continuous affinity distributions we did
not identify instances of true $WW$ outcomes that exhibited a decline
in adoption as $\rho$ increased beyond its ``threshold'' value.  This
is not unexpected, since as mentioned earlier, the shape of the joint
distribution and not just the correlation coefficient is expected to
affect the outcome.  Hence, as distributions change, so will the exact
configurations under which different effects arise as well as their
magnitude.  However, we believe that the general insight articulated
in the previous section still holds, namely, the presence of a minimum
correlation value to realize the critical mass of early adopters that
a bundle with a high externality factor needs to succeed, and the fact
that increasing $\rho$ beyond this value can narrow the bundle's
ultimate user base and, therefore, lower overall adoption unless its
externality factor is large enough.

\subsection{Summary}

Based on the above results, the following \emph{bundling guidelines}
emerge to assist in identifying services, which, if bundled, can result
in $WW$ outcomes:

\noindent
{\bf Bundling guidelines}: When bundling network services so as to
bolster their adoption, it is best to choose services that are
\begin{enumerate*}
\item 
  \begin{enumerate*}
  \item either heterogeneous in their cost-benefit structure, \ie low cost \&
    externality vs.~high cost \& externality,
 \item or of average cost and externality,
  \end{enumerate*}
\item and sufficiently correlated in how users value them, but not
  too much.
\end{enumerate*}
The first guideline highlights that successful bundling outcomes
require both a high overall externality factor, and a low enough cost
to allow the creation of a sufficient critical mass of early adopters
so that the value of the high externality can start being realized.
The second guideline states that creating a sufficient critical mass
of early adopters requires a certain minimum level of correlation in
how users value the bundled services, but that once this level has
been reached there is no benefit in selecting services that exhibit
higher levels of correlation (and there could be disadvantages).

\section{Conclusion}
\label{sec:con}

The paper presents an initial investigation aimed at developing a
better understanding of when bundling networking technologies or
services can be beneficial, \ie result in higher adoption levels than
when they are offered separately.

The question is of relevance in many practical settings as networking
technologies commonly face early adoption hurdles until they acquire a
large enough user-base to start delivering sufficient value.  Bundling
technologies can offer an effective solution to overcome those early
adoption challenges, but it is often hard to predict whether it will
succeed or not.  Of particular importance in determining the outcome
is correlation in how users value the individual technologies being
bundled.   The paper proposes simple models that can help explore this
question in a principled manner, and illustrates the type of insight
they provide through a few simple examples.

There are obviously many extensions that are desirable to the
basic models described in the paper and in their ability to
realistically capture how technologies interact, \eg the extent to
which they are complements or substitutes, or whether they exhibit
economies of scope.  The methodology outlined in the paper, however,
offers a first step towards developing a fundamental understanding of
the role that bundling can play in helping network technologies
overcome initial adoption hurdles.

% SW: added 10/16/13 to arXiv, not yet added to ACM-ToIT submission.
\section{Acknowledgements}
The authors wish to acknowledge the support of the National Science Foundation  through award CNS-1116039.  The contents of this paper belong to the authors and do not necessarily reflect the views of the National Science Foundation.  The authors also wish to acknowledge helpful discussions with Kartik Hosanagar.

%% References
\bibliographystyle{ACM-Reference-Format-Journals}
\bibliography{acmtoit2013}

%%% -*-BibTeX-*-
%%% Do NOT edit. File created by BibTeX with style
%%% ACM-Reference-Format-Journals [18-Jan-2012].

\begin{thebibliography}{00}

%%% ====================================================================
%%% NOTE TO THE USER: you can override these defaults by providing
%%% customized versions of any of these macros before the \bibliography
%%% command.  Each of them MUST provide its own final punctuation,
%%% except for \shownote{}, \showDOI{}, and \showURL{}.  The latter two
%%% do not use final punctuation, in order to avoid confusing it with
%%% the Web address.
%%%
%%% To suppress output of a particular field, define its macro to expand
%%% to an empty string, or better, \unskip, like this:
%%%
%%% \newcommand{\showDOI}[1]{\unskip}   % LaTeX syntax
%%%
%%% \def \showDOI #1{\unskip}           % plain TeX syntax
%%%
%%% ====================================================================

\ifx \showCODEN    \undefined \def \showCODEN     #1{\unskip}     \fi
\ifx \showDOI      \undefined \def \showDOI       #1{{\tt DOI:}\penalty0{#1}\ }
  \fi
\ifx \showISBNx    \undefined \def \showISBNx     #1{\unskip}     \fi
\ifx \showISBNxiii \undefined \def \showISBNxiii  #1{\unskip}     \fi
\ifx \showISSN     \undefined \def \showISSN      #1{\unskip}     \fi
\ifx \showLCCN     \undefined \def \showLCCN      #1{\unskip}     \fi
\ifx \shownote     \undefined \def \shownote      #1{#1}          \fi
\ifx \showarticletitle \undefined \def \showarticletitle #1{#1}   \fi
\ifx \showURL      \undefined \def \showURL       #1{#1}          \fi

\bibitem[\protect\citeauthoryear{Arthur}{Arthur}{2013}]%
        {arthur2013}
{C. Arthur}. 2013.
\newblock {NSA} scandal: what data is being monitored and how does it work?
\newblock The Guardian.   (2013).
\newblock
\newblock
\shownote{Available at
  http://www.guardian.co.uk/world/2013/jun/07/nsa-prism-records-surveillance-questions.}


\bibitem[\protect\citeauthoryear{Bakos and Brynjolfsson}{Bakos and
  Brynjolfsson}{1999}]%
        {BakosBrynjolfsson1999}
{Y. Bakos} {and} {E. Brynjolfsson}. 1999.
\newblock \showarticletitle{Bundling Information Goods: Pricing, Profits, and
  Efficiency}.
\newblock {\em Management Science\/} {45}, 12 (December 1999).
\newblock


\bibitem[\protect\citeauthoryear{Bhargava and Choudhary}{Bhargava and
  Choudhary}{2004}]%
        {Bhargava}
{H.~K. Bhargava} {and} {V. Choudhary}. 2004.
\newblock \showarticletitle{{Economics of an information intermediary with
  aggregation benefits}}.
\newblock {\em Information Systems Research.\/} {15}, 1 (2004), 22--36.
\newblock


\bibitem[\protect\citeauthoryear{Brewster}{Brewster}{2013}]%
        {brewster2013}
{T. Brewster}. 2013.
\newblock Tor network spike caused by botnet.
\newblock TechWeekEurope.   (September 2013).
\newblock
\newblock
\shownote{Available at
  http://www.techweekeurope.co.uk/news/mevade-botnet-tor-network-126497.}


\bibitem[\protect\citeauthoryear{Cabral}{Cabral}{1990}]%
        {Cabral1990}
{L.M.B. Cabral}. 1990.
\newblock \showarticletitle{On the adoption of innovation with network
  externalities}.
\newblock {\em Mathematical Social Sciences\/}  {19} (1990), 299--308.
\newblock


\bibitem[\protect\citeauthoryear{Chao and Derdenger}{Chao and
  Derdenger}{2013}]%
        {ChaoDerden2013}
{Y. Chao} {and} {T. Derdenger}. 2013.
\newblock \showarticletitle{Mixed bundling in two-sided markets in the presence
  of installed base effects}.
\newblock {\em Management Science\/} {57}, 3 (March 2013).
\newblock


\bibitem[\protect\citeauthoryear{{Fern\'andez Franco}}{{Fern\'andez
  Franco}}{2012}]%
        {franco2012}
{L. {Fern\'andez Franco}}. 2012.
\newblock A survey and comparison of anonymous communication systems: Anonymity
  and security.
\newblock Universitat Oberta de Catalunya - Institutional Repository.   (June
  2012).
\newblock
\newblock
\shownote{Available at http://hdl.handle.net/10609/14740.}


\bibitem[\protect\citeauthoryear{Fudenberg and Tirole.}{Fudenberg and
  Tirole.}{1991}]%
        {Fudenberg}
{D. Fudenberg} {and} {J. Tirole.} 1991.
\newblock {\em Game Theory}.
\newblock MIT Press., Cambridge, MA.
\newblock


\bibitem[\protect\citeauthoryear{Hotelling and Pabst}{Hotelling and
  Pabst}{1936}]%
        {HotPab1936}
{H. Hotelling} {and} {M.R. Pabst}. 1936.
\newblock \showarticletitle{Rank correlation and tests of significance
  involving no assumption of normality}.
\newblock {\em The Annals of Mathematical Statistics\/}  {7} (1936), 29--43.
\newblock


\bibitem[\protect\citeauthoryear{McAfee, McMillan, and Whinston}{McAfee
  et~al\mbox{.}}{1989}]%
        {McAfeeMcMillanWhinston1989}
{R.P. McAfee}, {J. McMillan}, {and} {M.D. Whinston}. 1989.
\newblock \showarticletitle{Multiproduct Monopoly, Commodity Bundling, and
  Correlation of Values}.
\newblock {\em The Quarterly Journal of Economics\/} {104}, 2 (May 1989).
\newblock


\bibitem[\protect\citeauthoryear{Nelsen}{Nelsen}{2009}]%
        {Nel2009}
{R.B. Nelsen}. 2009.
\newblock {\em An introduction to copulas\/} (second ed.).
\newblock Springer.
\newblock


\bibitem[\protect\citeauthoryear{Ozment and Schechter}{Ozment and
  Schechter}{2006}]%
        {OzmSch06}
{A. Ozment} {and} {S.E. Schechter}. 2006.
\newblock \showarticletitle{Bootstrapping the adoption of {I}nternet security
  protocols}. In {\em Proc. WEIS}.
\newblock


\bibitem[\protect\citeauthoryear{Pang and Etzion}{Pang and Etzion}{2012}]%
        {PanEtz2012}
{M.-S. Pang} {and} {H. Etzion}. 2012.
\newblock \showarticletitle{Analyzing pricing strategies for online services
  with network effects}.
\newblock {\em Information Systems Research\/} {23}, 4 (December 2012).
\newblock


\bibitem[\protect\citeauthoryear{Pearson}{Pearson}{1907}]%
        {Pea1907}
{K. Pearson}. 1907.
\newblock \showarticletitle{On further methods of determining correlation}.
\newblock {\em Drapers' Company Research Memoirs Biometric Series {IV}\/}
  (1907).
\newblock


\bibitem[\protect\citeauthoryear{Peres, Muller, and Mahajan}{Peres
  et~al\mbox{.}}{2010}]%
        {PerMul2010}
{R. Peres}, {E. Muller}, {and} {V. Mahajan}. 2010.
\newblock \showarticletitle{Innovation diffusion and new product growth models:
  A critical review and research directions}.
\newblock {\em International Journal of Research in Marketing\/}  {27} (2010).
\newblock


\bibitem[\protect\citeauthoryear{Prasad, Venkatesh, and Mahajan}{Prasad
  et~al\mbox{.}}{2010}]%
        {PraVen2010}
{A. Prasad}, {R. Venkatesh}, {and} {V. Mahajan}. 2010.
\newblock \showarticletitle{Optimal bundling of technological products with
  network externality}.
\newblock {\em Management Science\/} {56}, 12 (December 2010).
\newblock


\bibitem[\protect\citeauthoryear{Schmalensee}{Schmalensee}{1984}]%
        {Schmalensee1984}
{R. Schmalensee}. 1984.
\newblock \showarticletitle{Gaussian Demand and Commodity Bundling}.
\newblock {\em Journal of Business\/} {57}, 1 (1984).
\newblock


\bibitem[\protect\citeauthoryear{Venkatesh and Mahajan}{Venkatesh and
  Mahajan}{2009}]%
        {VenMah2009}
{R. Venkatesh} {and} {V. Mahajan}. 2009.
\newblock \showarticletitle{The design and pricing of product bundles: A review
  of normative guidelines and practical approaches}.
\newblock In {\em Handbook of Pricing Research in Marketing}, {V.R. Rao} (Ed.).
  Edward Elgar Publishing, 232--257.
\newblock


\end{thebibliography}

\clearpage
\appendix
\section*{APPENDIX}
\setcounter{section}{1}
\subsection{Generating and characterizing correlated random variables}
\label{sec:gen_corr_var}

The generation of a pair of correlated uniform random variables
$(U_1,U_2)$ is based on the following proposition.
\begin{proposition}[\cite{Pea1907,HotPab1936}] 
\label{prop:pearson}
Let $(Z_1,Z_2)$ be a pair of independent standard normal RVs and fix $\rho \in [-1,1]$.  Then
\begin{equation}
\left[ \begin{array}{c} Y_1 \\ Y_2 \end{array} \right] = \left[ \begin{array}{cc} Z_1 & Z_2 \end{array} \right] \left[ \begin{array}{cc} 1 & \rho \\ 0 & \sqrt{1-\rho^2} \end{array} \right] = \left[ \begin{array}{c} Z_1 \\ \rho Z_1 + \sqrt{1-\rho^2} Z_2 \end{array} \right]
\label{eq:c}
\end{equation}
are standard normal RVs with correlation $\rho$.  Further, $U = (U_1,U_2)$ with $U_i = F_Z(Y_i)$ for $i \in \{1,2\}$ are uniform RVs with correlation 
\begin{equation}
\rho_U = \frac{6}{\pi} \sin^{-1} \left( \frac{\rho}{2} \right) \in [-1,1].
\label{eq:d}
\end{equation}
\end{proposition}
\begin{remark}
Selecting $\rho = 2 \sin (\pi \rho^*/6)$ for a target correlation $\rho^*$ ensures $\rho_U = \rho^*$.  In what follows we will work with $\rho$ as the correlation parameter, even though $\rho_U$ is the actual correlation\footnote{A further justification for this equivocation is the fact that $\rho_U(\rho) \approx \rho$.  In fact $\max |\rho_U(\rho)-\rho|$ over $\rho \in [-1,1]$ occurs at $\rho_c = \pm \sqrt{4 \pi^2-36}/\pi \approx \pm 0.593664$ where $\rho_U(\rho_c) = \pm \frac{6}{\pi} \sin^{-1} \left(\sqrt{\pi^2-9}/\pi\right) \approx \pm 0.575581$, so the maximum deviation of $\rho_U(\rho)$ from $\rho$ is $|\rho_U(\rho_c) - \rho_c| \approx 0.01808$.}.
\end{remark}
\begin{remark}
Observe $A = \left[ \begin{array}{cc} 1 & \rho \\ 0 & \sqrt{1-\rho^2} \end{array} \right]$ in Prop.\ \ref{prop:pearson} is the Cholesky decomposition of the target correlation matrix $\Sigma = A^{\Tsf} A = \left[ \begin{array}{cc} 1 & \rho \\ \rho & 1 \end{array} \right]$.  
\end{remark}

From Prop.\ \ref{prop:pearson} it is immediate to obtain the joint CDF on $(U_1,U_2)$ in terms of the correlation $\rho$, and from there the joint PDF.

\begin{proposition}
\label{prop:jointcdfpdfcorrcont}
The joint CDF and joint PDF of $(U_1,U_2)$ in Prop.\ \ref{prop:pearson} at $(u_1,u_2)$ are:
\begin{eqnarray}
F_{U_1,U_2}(u_1,u_2) &=& \int_0^{u_1} F_Z \left( \frac{F_Z^{-1}(u_2) - \rho F_Z^{-1}(v_1)}{\sqrt{1-\rho^2}} \right) \drm v_1 \label{eq:jointcdfcorrcont} \\
f_{U_1,U_2}(u_1,u_2) &=& \frac{1}{\sqrt{1-\rho^2} f_Z(F_Z^{-1}(u_2))} f_Z \left( \frac{F_Z^{-1}(u_2) - \rho F_Z^{-1}(u_1)}{\sqrt{1-\rho^2}} \right) \label{eq:jointpdfcorrcont}.
\end{eqnarray}
\end{proposition}

\begin{proof}%[of Prop.\ \ref{prop:jointcdfpdfcorrcont}]
The joint CDF of $(U_1,U_2)$ at $(u_1,u_2)$ is:
\begin{eqnarray}
F_{U_1,U_2}(u_1,u_2) &=& \Pbb(U_1 \leq u_1, U_2 \leq u_2) \nonumber \\
&=& \Pbb(F_Z(Z_1) \leq u_1, F_Z(\rho Z_1 + \sqrt{1-\rho^2} Z_2) \leq u_2) \nonumber \\
&=& \Pbb(Z_1 \leq F_Z^{-1}(u_1), \rho Z_1 + \sqrt{1-\rho^2} Z_2 \leq F_Z^{-1}(u_2)) \nonumber \\
&=& \int_{-\infty}^{F_Z^{-1}(u_1)} \Pbb(Z_1 \leq F_Z^{-1}(u_1), \rho Z_1 + \sqrt{1-\rho^2} Z_2 \leq F_Z^{-1}(u_2) | Z_1 = z_1) f_Z(z_1) \drm z_1 \nonumber \\
&=& \int_{-\infty}^{F_Z^{-1}(u_1)} \Pbb(\rho z_1 + \sqrt{1-\rho^2} Z_2 \leq F_Z^{-1}(u_2)) f_Z(z_1) \drm z_1 \nonumber \\
&=& \int_{-\infty}^{F_Z^{-1}(u_1)} F_Z \left( \frac{F_Z^{-1}(u_2) - \rho z_1}{\sqrt{1-\rho^2}} \right) f_Z(z_1) \drm z_1
\end{eqnarray}
Change of variables from $z_1$ to $v_1 = F_Z(z_1)$ gives \Eqref{eq:jointcdfcorrcont}.  Differentiation w.r.t.\ $(u_1,u_2)$ gives
\begin{eqnarray}
\frac{\partial}{\partial u_1} F_{U_1,U_2}(u_1,u_2) 
&=& F_Z \left( \frac{F_Z^{-1}(u_2) - \rho F_Z^{-1}(u_1)}{\sqrt{1-\rho^2}} \right) \\
f_{U_1,U_2}(u_1,u_2) 
&=& \frac{1}{\sqrt{1-\rho^2}} f_Z \left( \frac{F_Z^{-1}(u_2) - \rho F_Z^{-1}(u_1)}{\sqrt{1-\rho^2}} \right) \frac{\drm}{\drm u_2} F_Z^{-1}(u_2) \nonumber 
\end{eqnarray}
Applying the inverse function theorem gives the joint PDF in \Eqref{eq:jointpdfcorrcont}. 
\end{proof}

The joint PDF $f_{U_1,U_2}(u_1,u_2)$ is illustrated in \fig{fig:jointpdfcorrcont} for $\rho \in \pm \frac{1}{2}$.  The following proposition shows that this joint distribution recovers the distributions of perfectly negatively correlated, independent, and perfectly positively correlated uniform random variables as $\rho \to \{-1,0,+1\}$, respectively.    

% SW (8/1/13): REPLACE with UNCOMPRESSED image for final version
% Original is 8.7MB, compressed is 176 KB
\begin{figure}[!ht]
\centering
\includegraphics[width=4.0in]{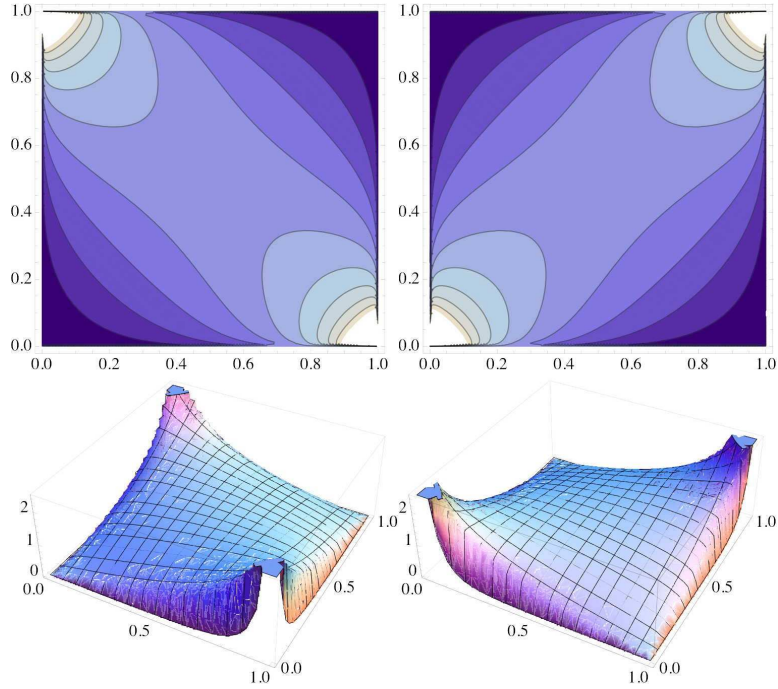}
\caption{Contour plots (top) and 3-D plots (bottom) of the joint PDF on service affinities $f_{U_1,U_2}(u_1,u_2)$ in Prop.\ \ref{prop:jointcdfpdfcorrcont} for $\rho = -1/2$ (left) and $\rho = 1/2$ (right).}
\label{fig:jointpdfcorrcont}
\end{figure}

\begin{proposition}
\label{prop:jointcdflimit}
The limits of $F_{U_1,U_2}(u_1,u_2)$ for $(u_1,u_2) \in [0,1]^2$ in \Eqref{eq:jointpdfcorrcont} as $\rho \to \{-1,0,+1\}$ are
\begin{eqnarray}
\lim_{\rho \to -1} F_{U_1,U_2}(u_1,u_2) &=& \max\{u_1+u_2-1,0\} \nonumber \\ 
\lim_{\rho \to 0} F_{U_1,U_2}(u_1,u_2) &=& u_1 u_2 \nonumber \\
\lim_{\rho \to 1} F_{U_1,U_2}(u_1,u_2) &=& \min\{u_1,u_2\}
\end{eqnarray}
corresponding to $U_1 + U_2 = 1$ (a.s.), $(U_1,U_2)$ independent, and $U_1 = U_2$ (a.s.), respectively.  
\end{proposition}

\begin{proof}%[of Prop.\ \ref{prop:jointcdflimit}]
Since the integral is over a finite support and the integrand is continuous:
\begin{equation}
\label{eq:conaffasympf}
\lim_{\rho \to \pm 1} F_{U_1,U_2}(u_1,u_2) = \int_0^{u_1} F_Z \left( \lim_{\rho \to \pm 1} \frac{F_Z^{-1}(u_2) - \rho F_Z^{-1}(v_1)}{\sqrt{1-\rho^2}} \right) \drm v_1.
\end{equation}
It follows that $F_{U_1,U_2}(u_1,u_2) \to u_1 u_2$ as $\rho \to 0$.  Next, as $\rho \to 1$ observe as $u_2 \gtrless v_1$ the function inside $F_Z(\cdot)$ in \Eqref{eq:conaffasympf} goes to $\pm \infty$, respectively.  Thus:
\begin{eqnarray}
u_1 < u_2 & \Rightarrow & \lim_{\rho \to 1} F_{U_1,U_2}(u_1,u_2) = \int_0^{u_1} F_Z(\infty) \drm v_1 = u_1 \nonumber \\
u_1 > u_2 & \Rightarrow & \lim_{\rho \to 1} F_{U_1,U_2}(u_1,u_2) = \int_0^{u_2} F_Z(\infty) \drm v_1 + \int_{u_2}^{u_1} F_Z(-\infty) \drm v_1 = u_2
\end{eqnarray}
and so $\lim_{\rho \to 1} F_{U_1,U_2}(u_1,u_2) = \min\{u_1,u_2\}$.  Finally, as $\rho \to -1$ observe as $F_Z^{-1}(u_2) + F_Z^{-1}(v_1) \gtrless 0$ the function inside $F_Z(\cdot)$ in \Eqref{eq:conaffasympf} goes to $\pm \infty$, respectively.  Observe
\begin{equation}
F_Z^{-1}(u_2) + F_Z^{-1}(v_1) \gtrless 0 ~ \Leftrightarrow ~ F_Z^{-1}(u_2)  \gtrless F_Z^{-1}(1-v_1) ~ \Leftrightarrow ~ v_1 + u_2 \gtrless 1.
\end{equation}
Thus:
\begin{eqnarray}
u_1 + u_2 \leq 1 & \Rightarrow & \lim_{\rho \to -1} F_{U_1,U_2}(u_1,u_2) = \int_0^{u_1} F_Z(-\infty) \drm v_1 = 0 \\
u_1 + u_2 \geq 1 & \Rightarrow & \lim_{\rho \to -1} F_{U_1,U_2}(u_1,u_2) = \int_0^{1-u_2} F_Z(-\infty) \drm v_1 + \int_{1-u_2}^{u_1} F_Z(\infty) \drm v_1 = u_1+u_2-1 \nonumber 
\end{eqnarray}
and so $\lim_{\rho \to 1} F_{U_1,U_2}(u_1,u_2) = \max\{u_1+u_2-1,0\}$.
\end{proof}
% SW: FIX THIS!!!  This should come out to 1_{u_1 + u_2 >= 1}.

The following two functions are central to the subsequent proposition. 
\begin{eqnarray}
\psi_{u,\rho}(v) & \equiv & F_Z \left( \frac{F_Z^{-1}(u-v) - \rho F_Z^{-1}(v)}{\sqrt{1-\rho^2}} \right) \label{eq:conaffbuncdfarg} \\
\phi_{u,\rho}(v) & \equiv & f_Z \left( \frac{F_Z^{-1}(u-v) - \rho F_Z^{-1}(v)}{\sqrt{1-\rho^2}} \right) \frac{1}{f_Z(F_Z^{-1}(u-v))} \label{eq:conaffbunpdfarg}
\end{eqnarray}
for $i)$ $u \in (0,2]$, $ii)$ $v \in [0,u]$ when $u \in (0,1]$ and $v \in [u-1,1]$ when $u \in (1,2]$, and $iii)$ $\rho \in [-1,1]$.  
\begin{proposition}
\label{prop:conaffsumpdfcdf}
The aggregate affinity under bundling $U = U_1 + U_2$ has CDF
\begin{equation}
F_U(u) = \left\{ \begin{array}{ll}
\displaystyle\int_0^u \psi_{u,\rho}(v) \drm v, \; & u \in (0,1] \\
\displaystyle\int_{u-1}^1 \psi_{u,\rho}(v) \drm v + u-1, \; & u \in [1,2]
\end{array} \right.
\label{eq:conaffsumcdf}
\end{equation}
and PDF
\begin{equation}
f_U(u) = \left\{ \begin{array}{ll}
\frac{1}{\sqrt{1-\rho^2}} \displaystyle\int_0^u \phi_{u,\rho}(v) \drm v, \; & u \in (0,1] \\
\frac{1}{\sqrt{1-\rho^2}} \displaystyle\int_{u-1}^1 \phi_{u,\rho}(v) \drm v, \; & u \in [1,2]
\end{array} \right.
\label{eq:conaffsumpdf}
\end{equation}
for $\psi_{u,\rho}(v)$ in \Eqref{eq:conaffbuncdfarg} and $\phi_{u,\rho}(v)$ in \Eqref{eq:conaffbunpdfarg}.  Further, $F_U(1) =1/2$ for all $\rho$.  
\end{proposition}

\begin{proof}[of Prop.\ \ref{prop:conaffsumpdfcdf}]
From \Eqref{eq:c}, the CDF of $U = U_1 + U_2$ in terms of the iid standard normal random variables $(Z_1,Z_2)$ and the correlation parameter $\rho$ is 
\begin{equation}
F_U(u) = \Pbb \left(F_Z(Z_1) + F_Z \left(\rho Z_1 + \sqrt{1-\rho^2} Z_2 \right) \leq u \right).
\end{equation}
For $u \in (0,1]$ condition on $z_1 \in \Rbb$, split the integral at $z_1 = F_Z^{-1}(u)$, and note the event of interest cannot occur for $z_1 > F_Z^{-1}(u)$:
\begin{equation}
F_U(u) = \int_{-\infty}^{F_Z^{-1}(u)} \Pbb(F_Z(Z_1) + F_Z(\rho Z_1 + \sqrt{1-\rho^2} Z_2) \leq u | Z_1 = z_1) f_Z(z_1) \drm z_1 
\end{equation}
Simplification gives the top equation in \Eqref{eq:conaffsumcdf1}.  For $u \in (1,2]$ condition on  $z_1 \in \Rbb$, split the integral at $z_1 = F_Z^{-1}(u-1)$ and notice the event of interest is assured for $z_1 \leq F_Z^{-1}(u-1)$:
\begin{equation}
F_U(u) 
= F_Z(F_Z^{-1}(u-1)) + \int_{F_Z^{-1}(u-1)}^{\infty} \Pbb(F_Z(z_1) + F_Z(\rho z_1 + \sqrt{1-\rho^2} Z_2) \leq u) f_Z(z_1) \drm z_1 
\end{equation}
Simplification gives the bottom equation in \Eqref{eq:conaffsumcdf1}.  
\begin{equation}
F_U(u) = \left\{ \begin{array}{ll}
\displaystyle\int_{-\infty}^{F_Z^{-1}(u)} F_Z \left( \frac{F_Z^{-1}(u-F_Z(z_1)) - \rho z_1}{\sqrt{1-\rho^2}} \right) f_Z(z_1) \drm z_1, \; & u \in (0,1] \\
u-1 + \displaystyle\int_{F_Z^{-1}(u-1)}^{\infty} F_Z \left( \frac{F_Z^{-1}(u-F_Z(z_1)) - \rho z_1}{\sqrt{1-\rho^2}} \right) f_Z(z_1) \drm z_1, \; & u \in [1,2]
\end{array} \right.
\label{eq:conaffsumcdf1}
\end{equation}
Change variables from $z_1$ to $v = F_Z(z_1)$ to obtain \Eqref{eq:conaffsumcdf}.

Apply the Leibniz integral rule to differentiate $F_U(u)$ and apply the inverse function theorem.  For $u \in [0,1]$:
\begin{eqnarray}
f_U(u) 
&=& F_Z \left( \frac{F_Z^{-1}(u-u) - \rho F_Z^{-1}(u)}{\sqrt{1-\rho^2}} \right) 
+ \int_0^u \frac{\drm}{\drm u} F_Z \left( \frac{F_Z^{-1}(u-v) - \rho F_Z^{-1}(v)}{\sqrt{1-\rho^2}} \right) \drm v \nonumber \\
&=& \frac{1}{\sqrt{1-\rho^2}} \int_0^u f_Z \left( \frac{F_Z^{-1}(u-v) - \rho F_Z^{-1}(v)}{\sqrt{1-\rho^2}} \right) \frac{1}{f_Z(F_Z^{-1}(u-v))} \drm v
\end{eqnarray}
Change variables from $z_1$ to $v = F_Z(z_1)$ to obtain the top equation in \Eqref{eq:conaffsumpdf}.  Likewise, for $u \in (1,2]$:
\begin{eqnarray}
f_U(u) 
&=& 1 - F_Z \left( \frac{F_Z^{-1}(u-(u-1)) - \rho F_Z^{-1}(u-1)}{\sqrt{1-\rho^2}} \right) + \nonumber \\
& & \int_{u-1}^1 \frac{\drm}{\drm u}  F_Z \left( \frac{F_Z^{-1}(u-v) - \rho F_Z^{-1}(v)}{\sqrt{1-\rho^2}} \right) \drm v \nonumber \\
&=& \frac{1}{\sqrt{1-\rho^2}} \int_{u-1}^1 f_Z \left( \frac{F_Z^{-1}(u-v) - \rho F_Z^{-1}(v)}{\sqrt{1-\rho^2}} \right) \frac{1}{f_Z(F_Z^{-1}(u-v))} \drm v
\end{eqnarray}
Change variables from $z_1$ to $v = F_Z(z_1)$ to obtain the bottom equation in \Eqref{eq:conaffsumpdf}.

Finally, we show $F_U(1) = 1/2$ for all $\rho$.  Observe $F_Z^{-1}(1-v) = - F_Z^{-1}(v)$ and thus
\begin{equation}
F_U(1) = \int_0^1 F_Z \left( \frac{F_Z^{-1}(1-v) - \rho F_Z^{-1}(v)}{\sqrt{1-\rho^2}} \right) \drm v = \int_0^1 F_Z \left( -F_Z^{-1}(v) \sqrt{\frac{1+\rho}{1-\rho}} \right) \drm v.
\end{equation}
Set $a = -\sqrt{(1+\rho)/(1-\rho)}$ and write the last expression above as 
\begin{equation}
g_1(a) = \int_0^1 F_Z(a F_Z^{-1}(v)) \drm v = \int_{-\infty}^{\infty} F_Z(a z) f_Z(z) \drm z,
\end{equation}
using the change of variable $z = F_Z^{-1}(v)$.  The derivative w.r.t.\ $a$ is
\begin{equation}
g_1'(a) = \int_{-\infty}^{\infty} z f_Z(az) f_Z(z) \drm z = \int_{-\infty}^{\infty} g_2(z,a) \drm z,
\end{equation}
for $g_2(z,a) \equiv  z f_Z(az) f_Z(z)$.  Now observe $g_2(z,a) = - g_2(-z,a)$, i.e., $g_2(z,a)$ is an odd function in $z$ for all $a$, and thus $g_1'(a) = 0$, and $g_1(a)$ is a constant for all $a$.   Using the change of variable $u = F_Z(z)$ at $a=1$ gives
\begin{equation}
g_1(1) = \int_0^1 u \drm u = \frac{1}{2}.
\end{equation}
Thus $F_U(1) = 1/2$ for all $\rho$.
\end{proof}

Representative plots of the CDF and PDF for $U$ in
Prop.\ \ref{prop:conaffsumpdfcdf} are shown in
\fig{fig:sumcdfcorrcont}.
%% We consider in the next three subsections the three special cases
%% of independence, perfect positive ($U_1=U_2$) and negative
%% ($U_1=1-U_2$) correlation, respectively.  

\begin{figure}[!ht]
\centering
\includegraphics[width=2.25in]{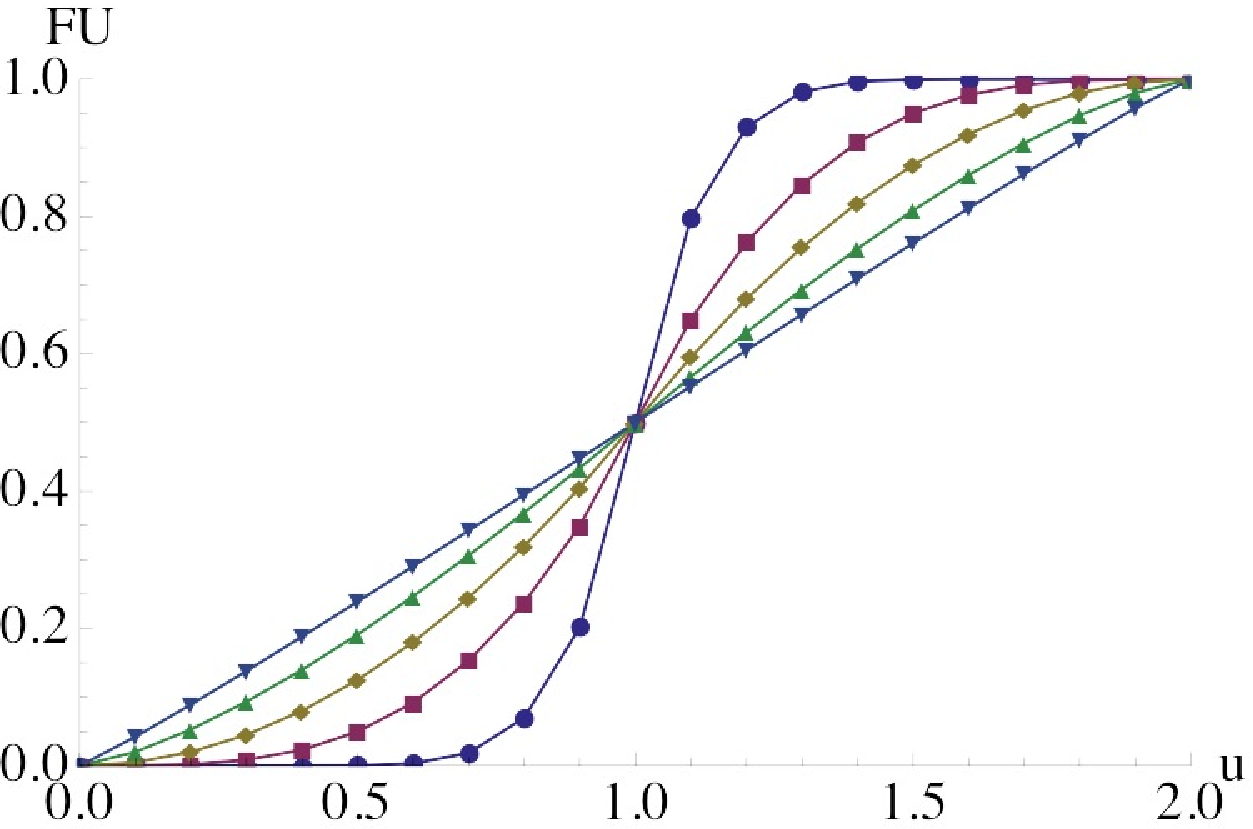}
\includegraphics[width=2.25in]{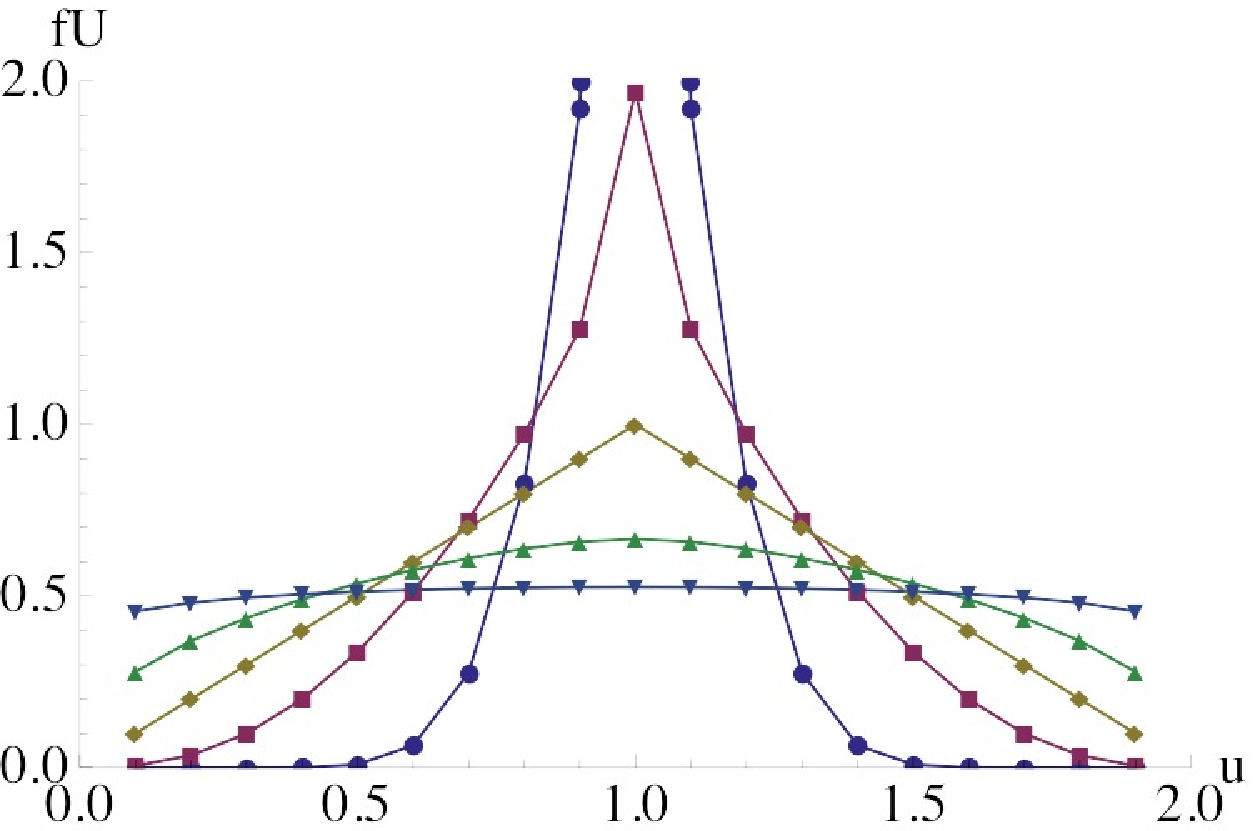}
\caption{The CDF (left) and PDF (right) for the aggregate affinity $U = U_1 + U_2$ from Prop.\ \ref{prop:conaffsumpdfcdf} for $\rho \in \{-0.9, -0.5,0,0.5,0.9\}$.}
\label{fig:sumcdfcorrcont}
\end{figure}

\begin{remark}
\label{rem:copula}
The standard tool to construct a joint distributions with specified
marginals is the copula \cite{Nel2009}.  In our context, a copula
would specify a joint distribution on $[0,1]^2$ with uniform marginals
on $[0,1]$.  Although there are many copulas that handle this quite
easily, our requirements are a bit specific in that we desire $i)$ to
directly parameterize the correlation of the joint distribution, and
$ii)$ to have a ``simple'' distribution for the sum $U = U_1+U_2$.
Although the construction we have employed falls short of this second
objective in that $F_U$ is expressible only in terms of an integral,
nonetheless our preliminary investigation into copulas has not
identified a candidate family of copulas meeting both objectives.   
\end{remark}

\begin{remark}
\label{rem:disbncorr}
Correlation is not in general a sufficient parameter to completely
capture the dependence of the adoption level on the joint
distribution.  In fact we expect that the adoption levels of two joint
distributions on $[0,1]^2$ with uniform marginals and common
correlation may have distinct adoption levels, precisely because the
solution of $h(x)=x$ depends upon the distribution of the aggregate
affinity, $F_U$.  Nonetheless, we view the correlation parameter as an
insightful knob to vary in order to highlight the fact that the
adoption level is quite sensitive to the joint distribution of the
affinities.   
\end{remark}

\subsection{Separate adoption equilibria under uniformly distributed
  user affinities}

Proposition~\ref{prop:conaffsepoffequ} characterized the possible
equilibria for separate service offerings when the user service
affinities are uniform random variables.
The results are illustrated in \fig{fig:conaffsep1}.
\begin{figure}[!ht]
\centering
\includegraphics[width=\textwidth]{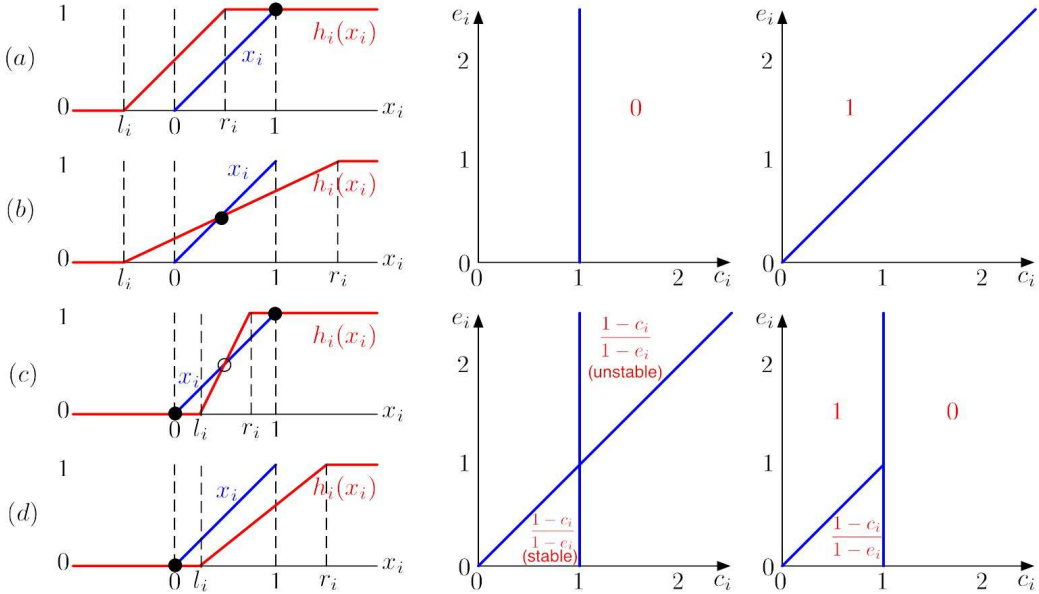}
\caption{Illustration of Prop.\ \ref{prop:conaffsepoffequ}.
  Equilibria adoption levels for continuous affinities with separate
  service offerings.  Equilibria $x_i^* \in [0,1]$ solve $h_i(x_i) =
  x_i$, where $h_i(x_i)$ has thresholds $l_i,r_i$.  The left figure
  shows the four orderings $(a)$ $l_i < 0 < r_i < 1$, $(b)$ $l_i < 0 <
  1 < r_i$, $(c)$ $0 < l_i < r_i < 1$, and $(d)$ $0 < l_i$ and $1 <
  r_i$, with black (open) dots indicating stable (unstable)
  equilibria.  The possible equilibria are $0,1,(1-c_i)/(1-e_i)$.  The
  first three subfigures of the right figure show the $(c_i,e_i)$
  plane and the regions for which each of the three equilibria are
  found.  The final subfigure on the bottom right shows a partition of
  the $(c_i,e_i)$ plane in terms of the lowest possible stable
  equilibria.  } 
\label{fig:conaffsep1}
\end{figure}
\begin{figure}[!h]
\centering
\includegraphics[width=\textwidth]{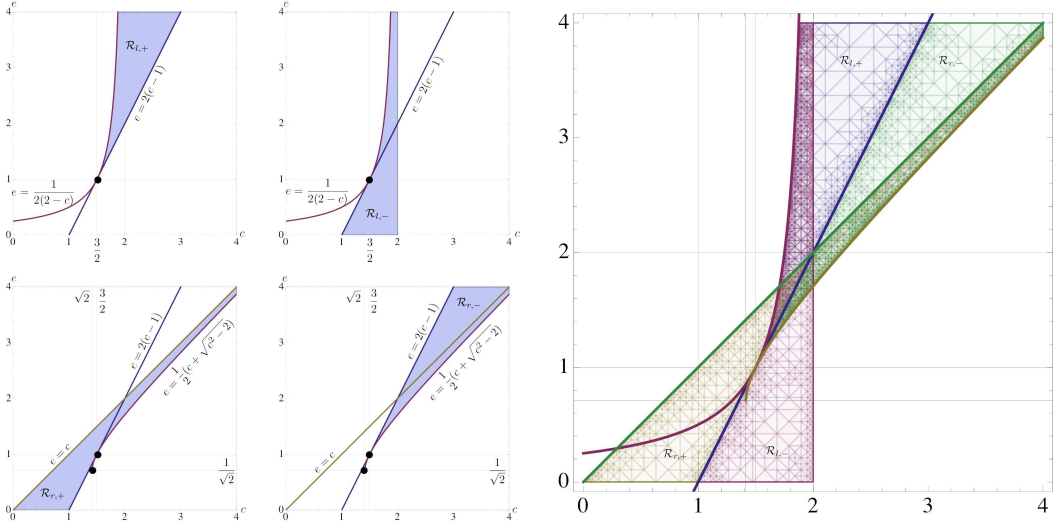}
\caption{The four equilibria regions $\Rmc_{l,\pm},\Rmc_{r,\pm}$ of
  the $(c,e)$ plane in \Eqref{eq:conaffindeqreg1} for continuous and
  independent affinities $(U_1,U_2)$.} 
\label{fig:conaffindeqregcomb1}
\end{figure}
\subsection{Bundle adoption equilibria under continuous users affinity
distributions}
 
\begin{proposition}
\label{prop:conaffbunprobadopt}
The probability of bundle adoption $h(x)$ in \Eqref{eq:BunCCDF} at adoption level $x$ for aggregate continuous affinity $U$ from Prop.\ \ref{prop:conaffsumpdfcdf} is
\begin{equation}
h(x) = \left\{ \begin{array}{lrcccl}
0, \; & & & x & \leq & l \\
2 - (c-ex) - \displaystyle\int_{c-ex-1}^1 F_Z \left( \frac{F_Z^{-1}(c-ex-v) - \rho F_Z^{-1}(v)}{\sqrt{1-\rho^2}} \right) \drm v , \; & l & < & x & \leq & m \\
1 - \displaystyle\int_0^{c-ex} F_Z \left( \frac{F_Z^{-1}(c-ex-v) - \rho F_Z^{-1}(v)}{\sqrt{1-\rho^2}} \right) \drm v, \; & m & < & x & \leq & r \\
1, \; & r & < & x & & 
\end{array} \right., 
\label{eq:conaffbunprobadopt}
\end{equation}
for adoption thresholds $l \equiv \frac{c-2}{e}$, $m \equiv \frac{c-1}{e}$, and $r \equiv \frac{c}{e}$.  The function $h(x)$ has the following properties:
\begin{enumerate}
\item $h'(x) = e f_U(c-ex) \geq 0$
\item $h^{''}(x) = -e^2 f_U'(c-ex)$
\item $h(m) = \frac{1}{2}$
\end{enumerate}
Stable equilibria include $x^* \in \{0,1\}$, where $x^* = 0$ is an equilibrium provided $h(0) = 0 \Leftrightarrow c > 2$, and $x^* = 1$ is an equilibrium provided $h(1) = 1 \Leftrightarrow c < e$.  
\end{proposition}

\begin{proof}
\Eqref{eq:conaffbunprobadopt} is immediate from the definition $h(x) = \Pbb(U > c - ex)$ in \Eqref{eq:BunCCDF} and Prop.\ \ref{prop:conaffsumpdfcdf}.  The first two properties are immediate from the definition of $h(x)$ and the fact that the CDF of $U$ is differentiable, by assumption.  The property $h(m) = 1/2$ follows immediately from $F_U(1) = 1/2$ in Prop.\ \ref{prop:conaffsumpdfcdf}.  
\end{proof}

As stated in Corollary~\ref{cor:conaffindprobadopt}, bundle adoption
equilibria satisfy $h(x)=x$ with solutions given by \Eqref{eq:conaffindeq1}.  

The regions in \Eqref{eq:conaffindeqreg1} are illustrated in
\fig{fig:conaffindeqregcomb1}, which illustrates their shape.
Observe $i)$ $e = 2(c-1)$ is the solution of $\xi_{l,\pm}^* = m$ and
$\xi_{r,\pm}^* = m$, $ii)$ $e = 1/(2(2-c))$ is the solution of
$2(c-2)e+1=0$ where $2(c-2)e+1$ is the discriminant of $\xi_{l,\pm}^*$
in \Eqref{eq:conaffindeq1}, and $iii)$ $e = \frac{1}{2}(c +
\sqrt{c^2-2})$ is the solution of $2(e-c)e+1=0$ where $2(e-c)e+1$ is
the discriminant of $\xi_{r,\pm}^*$ in \Eqref{eq:conaffindeq1}.

\subsection{Separate adoption equilibria under discrete user affinities}

\fig{fig:disaffsep} illustrates possible adoption equilibria
under discrete user affinities, and the regions of the $(c_i,e_i)$
plane they correspond to.
\begin{figure}[!h]
\centering
\includegraphics[width=\textwidth]{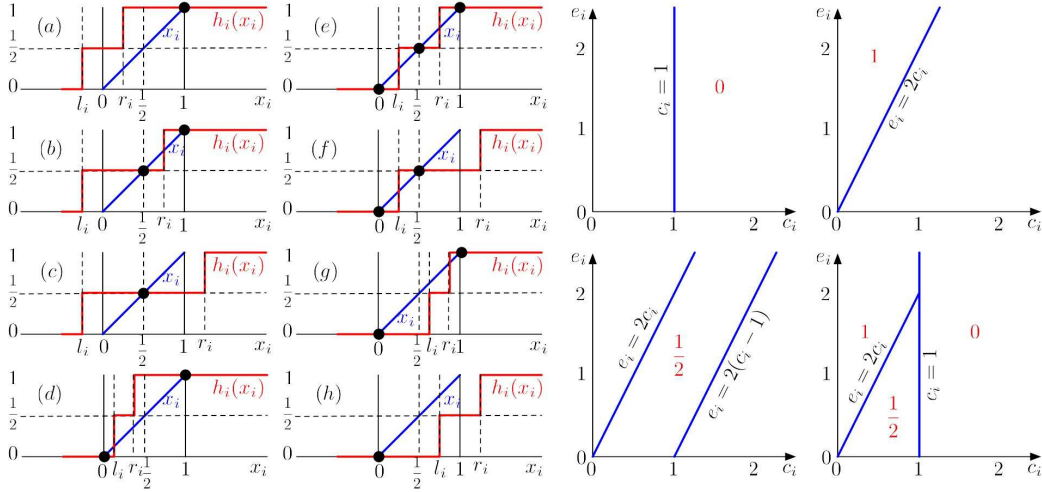}
\caption{
Illustration of Prop.\ \ref{pro:disaffsep}.  {\bf Left:} the eight
possible orderings of $\{l_i,r_i\}$ with $\{0,1/2,1\}$ each determine
the subset of $\{0,1/2,1\}$ that are equilibria. All equilibria are
stable.  {\bf Right:} the $(c_i,e_i)$ plane and the equilibria in each
region.  {\bf Bottom right:} partition of the $(c_i,e_i)$ plane
according to lowest stable equilibria.}   
\label{fig:disaffsep}
\end{figure}

\subsection{Bundle adoption equilibria under discrete user affinities} 

\fig{fig:disaffbun} parallels \fig{fig:disaffsep}, and
illustrates Prop.~\ref{pro:disaffbun}.  
Of interest is comparing the bottom-right plots of
\fig{fig:disaffsep} and \fig{fig:disaffbun}, to
identify the regions where bundling yields a higher adoption
equilibrium\footnote{Note though that the bottom-left plot of
\fig{fig:disaffbun} is for the specific value of $\rho=0$.}.
Different regions, and therefore outcomes arise as $\rho$ varies.
\begin{figure}[!h]
\centering
\includegraphics[width=\textwidth]{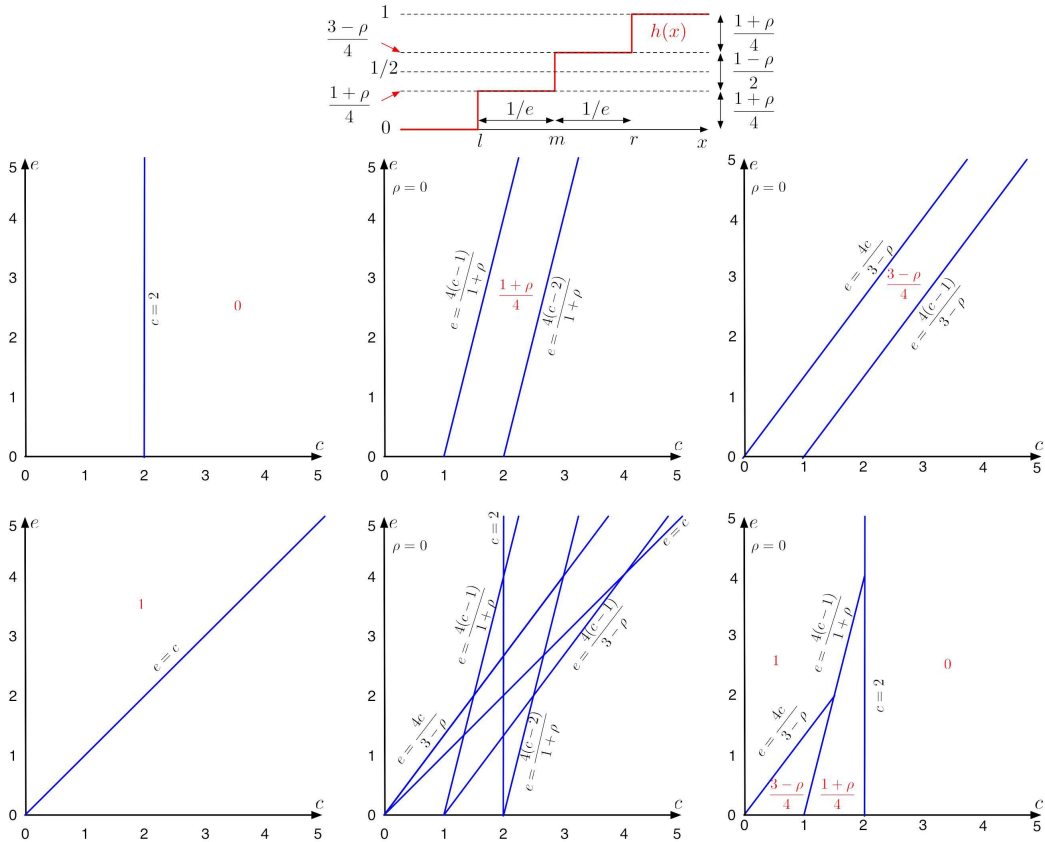}
\caption{
Illustration of Prop.\ \ref{pro:disaffbun}.  {\bf Top:} the
probability of adoption $h(x)$ in \Eqref{eq:disaffbunprobadopt} in
terms of $p$ (left) and $\rho$ (right).  {\bf Bottom:} the regions of
the $(c,e)$ plane for each of the four equilibria
$\{0,(1+\rho)/4,(3-\rho)/4,1\}$.  Also shown are the superimposed
boundaries of the four equilibria regions, as well as the partition of
the $(c,e)$ plane according to the lowest stable equilibria.  The
figures are shown for $\rho=0$.} 
\label{fig:disaffbun}
\end{figure}

\end{document}